\documentclass[10pt,a4paper]{article}
\usepackage{lipsum}
\usepackage{url}

\PassOptionsToPackage{utf8}{inputenc}
  \usepackage{inputenc}

\PassOptionsToPackage{T1}{fontenc} 
  \usepackage{fontenc}

\PassOptionsToPackage{
  drafting=false,    
  tocaligned=false, 
  dottedtoc=true,   
  eulerchapternumbers=false, 
  floatperchapter=true,      
  eulermath=false,  
  beramono=true,    
  palatino=true,    
  style=classicthesis 
}{classicthesis}

\newcommand{\myTitle}{Cohomology of psl(m|m) in trivial and in adjoint representations\xspace}
\newcommand{\myName}{Thiago Oliveira Ferreira\xspace}

\newcommand{\myDepartment}{Applied Mathematics\xspace}
\newcommand{\myUni}{University of Waterloo\xspace}

\newcommand{\myVersion}{\classicthesis}

\providecommand{\mLyX}{L\kern-.1667em\lower.25em\hbox{Y}\kern-.125emX\@}

\PassOptionsToPackage{canadian}{babel} 
    \usepackage{babel}

\usepackage{csquotes}
\PassOptionsToPackage{%
  backend=biber, 
  bibencoding=auto,
  language=auto,%
  style=numeric-comp,%
  giveninits=true, 
  sorting=none, 
  maxbibnames=10, 
  natbib=true, 
  urldate=long
}{biblatex}
    \usepackage{biblatex}

\PassOptionsToPackage{fleqn}{amsmath}       
  \usepackage{amsmath, amssymb}

\usepackage{graphicx} %
\usepackage{scrhack} 
\usepackage{xspace} 
\PassOptionsToPackage{printonlyused,smaller}{acronym}
  \usepackage{acronym} 

\usepackage{pgfplots} 

\usepackage{tabularx} 
\usepackage{subfig}

\usepackage{listings}
\PassOptionsToPackage{hyperfootnotes=false}{hyperref}
\usepackage{classicthesis}

\hypersetup{%
  colorlinks=true, linktocpage=true, pdfstartpage=3, pdfstartview=FitV,%
  breaklinks=true, pageanchor=true,%
  pdfpagemode=UseNone, %
  plainpages=false, bookmarksnumbered, bookmarksopen=true, bookmarksopenlevel=1,%
  hypertexnames=true, pdfhighlight=/O,
  urlcolor=CTurl, linkcolor=CTlink, citecolor=CTcitation, 
  pdftitle={\myTitle},%
  pdfauthor={\textcopyright\ \myName, \myUni, \myDepartment},%
  pdfsubject={},%
  pdfkeywords={},%
  pdfcreator={pdfLaTeX},%
  pdfproducer={LaTeX with hyperref and classicthesis}%
}

 \makeatletter
 \@ifpackageloaded{babel}%
   {%
     \addto\extrasamerican{%
     }%
     \addto\extrasngerman{%
     }%
     }{\relax}
 \makeatother

\usepackage[super]{nth}
\usepackage{amsthm, mathtools, upgreek}
\usepackage[align, nostrut]{tensor}

\usepackage{multirow, float}

\allowdisplaybreaks

\newtheorem{theorem}{Theorem}
\newtheorem{corollary}{Corollary}[theorem]
\newtheorem{proposition}{Proposition}
\newtheorem{lemma}{Lemma}

\theoremstyle{definition}
\newtheorem{definition}{Definition}

\makeatletter
\newtheorem*{rep@theorem}{\rep@title}
\newcommand{\newreptheorem}[2]{%
    \newenvironment{rep#1}[1]{%
     \def\rep@title{#2 \ref{##1}}%
     \begin{rep@theorem}}%
     {\end{rep@theorem}}}
\makeatother

\newreptheorem{theorem}{Theorem}

\usepackage{pgfplots, tikz}
\usepackage{tikz-cd}
    \pgfplotsset{compat=newest}
    \tikzcdset{arrow style=tikz, diagrams={>=Straight Barb} }
    \usetikzlibrary{babel}

\usepackage[nameinlink]{cleveref}
    \crefname{subsection}{subsection}{subsections}
    \Crefname{subsection}{Subsection}{Subsections}
\usepackage{float}

\newcommand{\sdots}{\mathbin{\reflectbox{$\ddots$}}}

\DeclareRobustCommand\longtwoheadrightarrow
     {\relbar\joinrel\twoheadrightarrow}
\DeclareRobustCommand\longhookrightarrow
     {\lhook\joinrel\longrightarrow}

\usepackage[labelfont=bf]{caption}
\usepackage[affil-it]{authblk}

\hypersetup{pdftitle    = {Cohomology of psl(n|n) in trivial and in adjoint representations},
            pdfauthor   = {{Thiago Oliveira Ferreira}},
            pdfsubject  = {},
            pdfkeywords = {}{}{}{},
			colorlinks  = true,
            }

\makeatletter
\def\@maketitle{%
  \newpage
  \null
  \vskip 2em%
  \begin{center}%
  \let \footnote \thanks
    {\Large\rmfamily\normalfont\spacedallcaps{\@title}\par}%
  \vskip 1.5em%
  { 
  \lineskip .5em%
  \begin{tabular}[t]{c}%
    \textsc{
    \@author%
    }
  \end{tabular}\par
  }%
  \vskip 1em%
  {\normalsize \@date}%
  \end{center}%
  \par
  \vskip 1.5em}
\makeatother

\begin{document}
\pagestyle{plain}

\title{Cohomology of $\mathfrak{psl}(n|n)$ in trivial and in adjoint representations}
\author[1,2,3]{\spacedlowsmallcaps{%
    Thiago~Oliveira~Ferreira%
    \thanks{\texttt{t6olivei@uwaterloo.ca}}}}
\affil[1]{Department of Applied Mathematics, University of Waterloo, Canada}
\affil[2]{Perimeter Institute for Theoretical Physics, Canada}
\affil[3]{Instituto de Física Teórica, São Paulo State University, Brazil}

\date{} 

\maketitle

\begin{abstract}
    \noindent
    We compute the cohomology of the projective special linear Lie superalgebra $\mathfrak{psl}(n|n)$ over a field of characteristic zero in trivial and in adjoint coefficients in full detail through its associated Hochschild--Serre spectral sequences. Applications in physics exist for its real form, the projective special unitary Lie superalgebra $\mathfrak{psu}(n|n)$ --- in particular, $\mathfrak{psu}(2,2|4)$ for AdS/CFT. We also find the cohomology of $\mathfrak{sl}(m|n)$ in trivial coefficients along the way and we point out a discrepancy with a previous statement in the literature. We present the explicit cohomology groups of low degree for $n\leqslant4$.
\end{abstract}

\tableofcontents

\section{Introduction} \label{sec: intro}

The interest in this topic started as a master's project for its applications to physics and, to the best of our knowledge, for the lack of a complete description of the cohomology of the projective special linear Lie superalgebra $\mathfrak{psl}(n|n)$ in the literature.

Besides the applications in physics (which are presented below), supermathematics is a rich field in itself. In particular, the cohomology of Lie superalgebras is an interesting topic: while finite dimensional Lie algebras always have finite dimensional cohomology, that is not the case for Lie superalgebras; the main object of this paper, $\mathfrak{psl}(n|n)$, is one example. This feature is only possible in the ``super'' setting because we allow for \emph{bosonic ghosts} (i.e. odd differential forms), which can tower up without restrictions --- unlike \emph{fermionic ghosts} (even differential forms).

As hinted above, we will borrow jargon from physics. Regarding the $\mathbb{Z}_2$-grading, elements of degree $0$ shall be called ``bosons'', while elements of degree $1$ shall be called ``fermions''. This comes from the supercommutator of an associative superalgebra $\mathfrak{g} = \mathfrak{g}_0 \oplus \mathfrak{g}_1$: on $\mathfrak{g}_0$ we have a commutator $[x,y] = xy-yx$, while on $\mathfrak{g}_1$ we have an anti-commutator $[u,v] = uv+vu$.

\subsection{AdS/CFT and Lie superalgebra cohomology}

The \emph{AdS/CFT correspondence} is one of the central themes of superstring theory. It relies on a duality between a Type IIB supergravity theory on a 5-dimensional anti-de Sitter space crossed with a 5-dimensional sphere ($\mathrm{AdS}_5 \times \mathbb{S}^5$) and a $\mathcal{N}=4$ supersymmetric Yang--Mills (SYM) theory on the boundary of that space \cite{Maldacena_1999, Berkovits_2005}. To accommodate supersymmetries, their gauge symmetries are encoded in a Lie \emph{super}group; namely the projective special unitary Lie supergroup $\mathrm{PSU}(2,2|4)$, whose Lie superalgebra is $\mathfrak{psu}(2,2|4)$.

One of the most straightforward questions one can ask is about the linearized deformations of the $\text{AdS}_5 \times \mathbb{S}^5$ background, which are equivalent to the linearized deformations of $\mathcal{N} = 4$ SYM. They form a representation of the algebra of supersymmetries $\mathfrak{psu}(2,2|4)$ of $\mathrm{AdS}_5 \times \mathbb{S}^5$. This representation is not irreducible, as there are invariant subspaces. Some invariant subspaces have a well-defined Hermitian scalar product;
they are unitary representations. Those subspaces have been well studied in the literature \cite{Gunaydin_2000}.

However, the whole space of deformations is non-unitary, and its structure has not been studied sufficiently well. For example, there are finite-dimensional subspaces, and their structure has only been conjectured \cite{Mikhailov:2011af}.

In the case of $\mathrm{AdS}_3$, the deformation representations were described in \cite{Troost:2011fd, Gaberdiel:2011vf}. They have a rich mathematical structure. But, in the case of $\mathrm{AdS}_5$, such a description is not known.

\bigskip

\emph{What is the cohomology of the supersymmetry algebra $\mathfrak{psu}(2,2|4)$ with coefficients in these representations?} That is a natural question, since cohomology is a natural generalization of invariants. It is a good ``probe'' into the structure of representations. But, besides that, the cohomologies of these representations and their tensor products can be potentially useful for computing the boundary $S$-matrix; that was outlined in \cite{Mikhailov:2019kfm, Mikhailov:2024sef}.

In the pure spinor approach to superstring theory \cite{Berkovits_2009}, the space of linearized deformations is itself a cohomology; namely, of a complex that looks similar to the Lie superalgebra complex. It differs by imposing constraints on the ghosts. The two differentials anticommute, forming a bicomplex. The matching of the two spectral sequences of this bicomplex imposes nontrivial restrictions on the dimensions of the cohomology groups, thus restricting the possible representations. Those restrictions just started to be unveiled in \cite{Ferreira-Mikhailov_2026}.

For this research program, we need to learn how to compute the cohomology spaces of $\mathfrak{psu}(2,2|4)$ with coefficients in various representations. Even for finite-dimensional representations, this is a highly nontrivial linear algebra problem. In this paper, we will compute the cohomology with coefficients in the trivial representation and in the adjoint representation; they are especially important for AdS/CFT as they correspond to BRST cohomology with ghost numbers $0$ and $1$, respectively \cite{Vallilo_2004}.

\subsection{Structure of this paper}

Instead of directly computing the cohomology of $\mathfrak{psu}(2,2|4)$ --- or $\mathfrak{psu}(n|n)$ for a general $n\in\mathbb{N}$ ---, we step back and consider its complexification
\begin{equation} \label{1 psl and psu}
    \mathfrak{psl}_\mathbb{C}(n|n)
        \cong \mathbb{C} \, \mathfrak{psu}(n)
        \coloneqq \mathfrak{psu}(n|n) \otimes_{\mathbb{R}} \mathbb{C} .
\end{equation}

That is because their cohomologies are isomorphic, in the sense that
\begin{equation*}
    H\bigl( \mathfrak{psl}_{\mathbb{C}}(n|n) ; V_{\mathbb{C}} \bigr)
        \cong \mathbb{C} \, H\bigl( \mathfrak{psu}(n|n) , V \bigr) ,
\end{equation*}
where $V_{\mathbb{C}}$ is the complexification of the real module $V$ \cite{Weibel_1994}. Besides, $\mathfrak{psl}(n|n)$ can be easily defined from $\mathfrak{gl}(n|n)$ through the restriction to the kernel of a map (the supertrace) and a quotient by an ideal, as described in \Cref{subsec: super}. Our computation does not rely on the field of complex numbers, so we state our results for any complete and infinite field $\mathbb{K}$ of characteristic zero.

\bigskip

The main results of this paper are the following:

\begin{reptheorem}{thm: psl trivial coefficients}[\cref{subsec: HS trivial coefficients}] \it
The cohomology of $\mathfrak{psl}(n|n)$ in trivial coefficients is given by
\begin{equation*}
    H \bigl( \mathfrak{psl}(n|n) \bigr)
        \cong H(\mathfrak{sl}_n) \otimes \mathbb{K}\bigl[ \eta_2 ,\ \delta^+ ,\ \delta^- \bigr] \Big/ \bigl( \delta^+ \, \delta^- = \tfrac{1}{n!} (\eta_2)^n \bigr) ,
\end{equation*}
where $\delta^\pm = \det(\gamma^\pm)$ are the determinant of the fermionic blocks $\gamma^\pm$ of $\mathfrak{psl}(n|n)$, and $\eta_2 = \operatorname{tr}(\gamma^+ \gamma^-)$ is the trivial central charge associated with $\mathfrak{sl}(n|n)$.
\end{reptheorem}

\begin{reptheorem}{thm: psl adjoint coefficients}[\cref{subsec: 2nd page HS adjoint}] \it
The cohomology of $\mathfrak{psl}(n|n)$ in adjoint coefficients is module over $\mathcal{K}$,
\begin{align*}
    H_{\operatorname{ad}} \bigl( \mathfrak{psl}(n|n) \bigr)
    = \mathcal{K}\bigl\{ (\gamma^+ - \gamma^-) ,\ \mu^+ ,\ \mu^- \bigr\} \big/ \mathcal{J}_{\operatorname{ad}} ,
\end{align*}
where $\mu^\pm$ are the adjugate matrices of $\gamma^\pm$, $\mathcal{K} \coloneqq H\bigl( \mathfrak{psl}(n|n) \bigr)$ is the cohomology in trivial coefficients, and $\mathcal{J}_{\operatorname{ad}}$ is the ideal generated by the relation
\begin{align*} 
    \delta^+ \, \mu^- - \delta^- \, \mu^+
        &= \frac{1}{(n-1)!} \, (\eta_2)^{n-1} \bigl( \gamma^+ - \gamma^- \bigr) .
\end{align*}
\end{reptheorem}

Their proofs use the Hochschild--Serre spectral sequences associated with the pair $\bigl( \mathfrak{psl}(n|n) , \mathfrak{sl}_n \oplus \mathfrak{sl}_n \bigr)$, where $\mathfrak{sl}_n \oplus \mathfrak{sl}_n \cong \mathfrak{psl}(n|n)_0$ is the bosonic subspace of $\mathfrak{psl}(n|n)$. We also present a few examples for $n \leqslant 4$.

\bigskip

To make this paper accessible to both the mathematics and the physics communities, we try to make it as self-contained as possible.
\begin{itemize}
    \item For that reason, in \Cref{sec preliminaries}, we introduce the preliminary definitions and well-established results that will be used later --- namely, definitions in superanalysis, the BRST approach to Lie superalgebra cohomology \emph{à la} AKSZ \cite{AKSZ_1995}, and the Hochschild--Serre spectral sequence. We also revise the cohomology of the general linear Lie superalgebra $\mathfrak{gl}(m|n)$ as done by \textcite{Fuks_1986}.

    \item In \Cref{sec: trivial coefficients}, we perform the computation of the cohomology of $\mathfrak{psl}(n|n)$ in trivial coefficients, invoking the fundamental theorems of invariant theory for $\mathrm{SL}_n$, and give the result for $\mathfrak{sl}(m|n)$ \emph{en passant}. We present an explicit list of generators of low order cohomology spaces for $n\leqslant4$.

    \item In \Cref{sec: adjoint coefficients}, we do the calculation in adjoint coefficients for $\mathfrak{psl}(n|n)$ alone and present a list of generators of low order cohomology spaces for $n=2$.
\end{itemize}

\subsection*{Notation}

We use ``$\coloneqq$'' for definitions and ``$\equiv$'' for equivalent notations. Einstein summation convention is assumed everywhere.

The field $\mathbb{K}$ will be taken to be complete and infinite of characteristic zero. Vector spaces and Lie algebras will be of finite dimension over $\mathbb{K}$. The tensor product of graded spaces is assumed to be graded naturally, viz.
\begin{align*}
    \deg(a \otimes b)
        \coloneqq \deg(a) + \deg(b) .
\end{align*}

The ring of polynomials on the coordinates $x_1,\ldots,x_p$ of a vector space $U$ with coefficients in $\mathbb{K}$ is, as usual, denoted by $\mathbb{K}[x_1,\ldots,x_p] \equiv \mathbb{K}[U] \equiv \mathrm{S}^\bullet(U^\ast)$. If $v_1,\ldots,v_k$ are elements of a module $V$ over the ring $R$, we shall denote by
\begin{equation*}
    R\{ v_{i_1} , \ldots , v_{i_q} \}
        \coloneqq \biggl\{ \sum_{k = 1}^q \alpha_k v_{i_k} \, \bigg| \, \alpha_1,\ldots,\alpha_q \in R \biggr\}
\end{equation*}
the linear span of $v_{i_1} , \ldots, v_{i_q}$ over $R$. Whenever it makes sense, the notation $a \cdot X$ means the elements of the form $a \cdot x$ for $x\in X$, where $\cdot$ is the underlying product. Group actions will be denoted by $\vartriangleright$.

All propositions and theorems without references are, to the best of our knowledge, new --- and, for that reason, followed by proofs.
\section{Preliminaries}
\label{sec preliminaries}

In this Section, we lay out the notation and definitions to be used later.
\begin{itemize}
    \item \Cref{subsec: super} is devoted to defining the Lie superalgebras we are interested in: $\mathfrak{gl}(m|n)$, $\mathfrak{sl}(m|n)$ and $\mathfrak{psl}(n|n)$;
    
    \item \Cref{subsec: BRST} introduces ghost coordinates and the BRST cohomology of Lie (super)algebras;
    
    \item \Cref{subsec: spectralseq} defines the Hochschild--Serre spectral sequence and its associated filtration; and
    
    \item \Cref{subsec: H gl trivial} recalls the computation of $H^\bullet\bigl(\mathfrak{gl}(m|n);\mathbb{K}\bigr)$ through its associated Hochschild--Serre spectral sequence, as found in \textcite{Fuks_1986}.
\end{itemize}

\subsection{Lie superalgebras} \label{subsec: super}

The definitions in this Subsection are due to \textcite{Berezin-Kats_1970, Kac_1977, Fuks_1986}.

As usual, a \emph{supervector space} is a $\mathbb{Z}_2$-graded vector space, a \emph{superalgebra} is a $\mathbb{Z}_2$-graded algebra, and so on. As mentioned in \Cref{sec: intro}, homogeneous elements of degree $0$ will be called ``bosonic'' (or ``even''); homogeneous elements of degree $1$ will be called ``fermionic'' (or ``odd''). The degree of an homogeneous element $X$ will be denoted $\deg(X)$ and, whenever $\deg(X)$ is written, it is assumed that $X$ is homogeneous.

\begin{definition}
A \emph{Lie superalgebra} $\mathfrak{g}$ is a superalgebra $\mathfrak{g} = \mathfrak{g}_0 \oplus \mathfrak{g}_1$ equipped with a binary operation $[\cdot,\cdot] \colon \mathfrak{g} \times \mathfrak{g} \longrightarrow \mathfrak{g}$, called the \emph{Lie superbracket}, such that
\begin{alignat}{3}
    \label{2 superbracket antisymmetry}
    [X,Y] &= -(-1)^{\deg(X)\deg(Y)} [Y,X] ,
    \qquad&\forall X,Y\in\mathfrak{g} ,
    \\
\label{2 superbracket Jacobi}
    \bigl[ X , [Y,Z] \bigr]
        &= \bigl[ [X,Y],Z \bigr]
        + (-1)^{\deg(X)\deg(Y)}\bigl[ Y,[X,Z] \bigr] ,
    \quad&\forall X,Y,Z\in\mathfrak{g} .
\end{alignat}
\end{definition}

The restriction of \eqref{2 superbracket antisymmetry} and \eqref{2 superbracket Jacobi} to $\mathfrak{g}_0 $ defines a Lie algebra structure, so $\mathfrak{g}_0 \longhookrightarrow \mathfrak{g}$ is an inclusion of Lie superalgebras. In turn, $\mathfrak{g}_1$ is a $\mathfrak{g}_0$-module whose $\mathfrak{g}_0$-action is given by the Lie superbracket.

On an associative superalgebra $\mathcal{A}$, the \emph{supercommutator}
\begin{equation} \label{2 supercommutator}
    [X,Y] = XY - (-1)^{\deg(X) \deg(Y)} YX
\end{equation}
turns $\mathcal{A}$ into a Lie superalgebra.

Take $m,n\in\mathbb{N}$ and let
\begin{align*}
    V_0 \coloneqq \mathbb{K}^m ,
    \qquad\text{and}\qquad
    V_1 \coloneqq \mathbb{K}^n .
\end{align*}

\begin{definition}
The \emph{general linear} Lie superalgebra $\mathfrak{gl}(V) \equiv \mathfrak{gl}(m|n)$ (with the field $\mathbb{K}$ implicit) is the associative superalgebra of endomorphisms of ${ V = V_0 \oplus V_1 }$ equipped with the supercommutator. Its homogeneous spaces are
\begin{subequations} 
\begin{align}
    \mathfrak{gl}(m|n)_0
        &\cong \mathfrak{gl}_m \oplus \mathfrak{gl}_n ,
    \\
    \mathfrak{gl}(m|n)_1
        &\cong \mathrm{Hom}(V_0,V_1) \oplus \mathrm{Hom}(V_1,V_0) .
        \label{2 gl m n 1}
\end{align}
\end{subequations}
\end{definition}

For what follows, it will be handy to define the shorter notation
\begin{subequations} \label{2 index structure of U}
\begin{align}
    U^+ &\coloneqq \mathrm{Hom}(V_0,V_1)
        \cong \mathrm{M}_{n \times m}(\mathbb{K}) ,
    \\
    U^- &\coloneqq \mathrm{Hom}(V_1,V_0)
        \cong \mathrm{M}_{m\times n}(\mathbb{K}) .
\end{align}
\end{subequations}
An element in $U^\pm$ is a linear transformation that increases/decreases the degree of a homogeneous supervector in $V$.

An element of $\mathfrak{gl}(m|n)$ can be put in the matrix form
\begin{equation} \label{2 gl decomposition}
    X =
    \begin{pmatrix}
        x & u^- \\ u^+ & y
    \end{pmatrix} ,
    \qquad\text{where}\quad
    \left\{
    \begin{aligned}
        X_0 &\equiv x \oplus y
            &&\in \mathfrak{gl}_m \oplus \mathfrak{gl}_n ,
        \\
        u   &\equiv u^+ \oplus u^-
            &&\in U^+ \oplus U^- .
    \end{aligned}
    \right.
\end{equation}
Beware that \eqref{2 gl decomposition} is simply a decomposition; it is \emph{not} a Lie (super)algebra homomorphism between $\mathfrak{gl}(m|n)$ and $\mathfrak{gl}_{m+n}$, since it clearly does not respect \eqref{2 supercommutator}.

\begin{definition}
The \emph{supertrace} is a map $\operatorname{str}\colon\mathfrak{gl}(m|n) \longrightarrow \mathbb{K}$ such that, for $X$ in the standard decomposition \eqref{2 gl decomposition}, it is given by
\begin{equation} \label{2 def supertrace}
    \operatorname{str}(X)
        \coloneqq \operatorname{tr}(x) - \operatorname{tr}(y) .
\end{equation}
This map does not depend on the choice of basis \cite{Kac_1977}. The \emph{special linear} Lie superalgebra $\mathfrak{sl}(m|n)$ is the kernel of the supertrace.
\end{definition}

Decomposition \eqref{2 gl decomposition} can be carried over to $\mathfrak{sl}(m|n)$, and the map
\begin{subequations} \label{2 isomorphism sl m n 0}
\begin{align}
        \mathfrak{sl}_m \oplus \mathfrak{sl}_n \oplus \mathfrak{u}_1
        &\longrightarrow \mathfrak{sl}(m|n)_0
        \label{2 sl m n 0}
        \\
        (x,y,\lambda) &\longmapsto
        \begin{pmatrix}
            x + \frac{\lambda}{m} \mathbb{1}_m & 0 \\
            0 & y + \frac{\lambda}{n} \mathbb{1}_n
        \end{pmatrix}
\end{align}
\end{subequations}
is an isomorphism of Lie algebras, where $\mathbb{1}_m$ is the $m\times m$ identity matrix. In general, the $1$-dimensional subalgebra $\mathfrak{u}_1$ acts non-trivially on $\mathfrak{sl}(m|n)_1$, since
\begin{align*}
    \operatorname{ad}_{(0,0,1)}
    \begin{pmatrix}
        0 & u^- \\ u^+ & 0
    \end{pmatrix}
        &= \biggl( \frac{1}{m} - \frac{1}{n} \biggr) \,
        \begin{pmatrix}
            0 & -u^+ \\ +u^- & 0
        \end{pmatrix} .
\end{align*}
It is clear that, when $m=n$, the action of $\mathfrak{u}_1$ on $\mathfrak{sl}(n|n)_1$ is trivial: indeed, $\mathfrak{u}_1$ becomes the central ideal generated by the $2n\times2n$ identity $(\mathbb{1}_n,\mathbb{1}_n) = \mathbb{1}_{2n}$.

\begin{definition}
The \emph{projective special linear} Lie superalgebra $\mathfrak{psl}(n|n)$ is the quotient of $\mathfrak{sl}(n|n)$ by the ideal $\mathfrak{i}_{2n} \coloneqq \mathbb{K}\{ \mathbb{1}_{2n} \}$.
\end{definition}

It is immediate to see that
\begin{subequations} \label{2 homogeneous subspaces of psl}
\begin{align}
    \mathfrak{psl}(n|n)_0
        &\cong \mathfrak{sl}_n \oplus \mathfrak{sl}_n ,
        \label{2 psl m n 0}
    \\
    \mathfrak{psl}(n|n)_1
        &\cong U^+ \oplus U^- .
\end{align}
\end{subequations}
The fermionic subspaces of $\mathfrak{gl}(m|n)$ and $\mathfrak{sl}(m|n)$ --- and, for $m=n$, $\mathfrak{psl}(n|n)$ --- are the same. We will write $U = U^+ \oplus U^-$ to indicate the space of odd linear transformations as $\mathfrak{g}_0$-module for the appropriate bosonic part $\mathfrak{g}_0$.

\bigskip

When decomposing $\mathfrak{gl}_n \cong \mathfrak{sl}_n \oplus \mathbb{K}$ as vector spaces, the subspace of $\mathfrak{gl}(m|n)_0$ that is absent in $\mathfrak{sl}(m|n)_0$ is the one spanned by
\begin{equation} \label{2 def J generator of supertrace}
    \mathbb{J}
        \coloneqq \begin{pmatrix}
            +\mathbb{1}_m & 0 \\ 0 & -\mathbb{1}_n
        \end{pmatrix} .
\end{equation}
By the definition of $\mathfrak{sl}(m|n)$ and by later discussion, $\mathbb{J}$ is the element of $\mathfrak{gl}(m|n)$ associated\footnote{More precisely, the \emph{ghost} of $\mathbb{J}$ is the supertrace map \eqref{2 def supertrace}; see \Cref{subsec: BRST}.} to the supertrace \eqref{2 def supertrace} and is the generator of $H^1\bigl( \mathfrak{gl}(m|n) \bigr)$.

\subsection{BRST cohomology of Lie superalgebras} \label{subsec: BRST}

The Lie algebra cohomology was introduced by \textcite{Chevalley-Eilenberg_1948} and can be adapted to Lie superalgebras straightforwardly.

\begin{definition}
The \emph{Chevalley--Eilenberg complex} of cochains of a Lie (super)algebra $\mathfrak{g}$ over $\mathbb{K}$ is
\begin{equation*}
    \mathrm{CE}^\bullet (\mathfrak{g})
        \coloneqq \mathrm{Hom}\bigl( \Lambda^\bullet (\mathfrak{g}) , \mathbb{K} \bigr)
        \cong \Lambda^\bullet (\mathfrak{g}^\ast) ,
\end{equation*}
equipped with the differential $d_{\mathrm{CE}}$ which is dual to the Lie (super)bracket of $\mathfrak{g}$,
\begin{equation*}
    \bigl\langle d_{\mathrm{CE}}(\xi) \,\big|\, X \wedge Y \bigr\rangle
        = -(-1)^{\deg(\xi) \deg(X)} \bigl\langle \xi \,\big|\, [X,Y] \bigr\rangle ,
        \qquad \forall \xi\in\mathfrak{g}^\ast ,\ \forall X,Y\in \mathfrak{g} ,
\end{equation*}
and is extended to $\Lambda^\bullet(\mathfrak{g}^\ast)$ as a derivation. The \emph{Chevalley--Eilenberg cohomology} of $\mathfrak{g}$ is the cohomology of the cochain complex $\bigl( \mathrm{CE}^\bullet(\mathfrak{g}) , d_{\text{CE}} \bigr)$.
\end{definition}

Instead, we can use the language of BRST operator and Faddeev--Popov ghosts \emph{à la} AKSZ \cite{AKSZ_1995}, which is more familiar to the physicists --- and will be used in this paper.

\begin{definition}
Let $V = V_0 \oplus V_1$ be a supervector space. Its \emph{statistics-reversal} is the supervector space $\Pi V$ whose homogeneous subspaces are
\begin{align*}
    (\Pi V)_0 &\coloneqq V_1 ,
    \\
    (\Pi V)_1 &\coloneqq V_0 .
\end{align*}
Elements of $\Pi V$ are called ``ghosts'', so that elements of $(\Pi V)_0$ and $(\Pi V)_1$ are \emph{fermionic ghosts} and \emph{bosonic ghosts}, respectively.
\end{definition}

Let $\Pi\mathfrak{g}$ be the statistics-reversal of $\mathfrak{g}$. The decomposition \eqref{2 gl decomposition} of $\mathfrak{gl}(m|n)$ can be carried over to $\Pi\mathfrak{gl}(m|n)^\ast$ as
\begin{equation} \label{2 gl ghost decomposition}
    c =
    \begin{pmatrix}
        a & \gamma^- \\ \gamma^+ & b
    \end{pmatrix} ,
    \qquad\text{where}\quad
    \left\{
    \begin{aligned}
        c_0
            &\equiv a \oplus b
            &&\in \bigl( \Pi\mathfrak{gl}(m|n)^\ast \bigr){}_1 ,
        \\
        \gamma
            &\equiv \gamma^+ \oplus \gamma^-
            &&\in \bigl( \Pi \mathfrak{gl}(m|n)^\ast \bigr){}_0 .
    \end{aligned}
    \right.
\end{equation}
We will see $a$, $b$ and $\gamma^\pm$ as coordinate functions on $\mathfrak{g}$, as in $\gamma^\pm \colon X \longmapsto u^\pm$ etc.

Denote $\mathrm{Pol}^\bullet(\Pi\mathfrak{g})$ the \emph{algebra of ghost polynomials} of $\mathfrak{g}$.

Let $\{ \tau\indices{_B} \}_B$ be a homogeneous basis of $\mathfrak{g}$, with structure constants $f\indices*{^A_B_C}$, and $\{ c\indices{^A} \}_A \subset \Pi\mathfrak{g}^\ast$ be the \nth{1} order ghost monomials generating $\mathrm{Pol}^\bullet (\Pi\mathfrak{g})$ as a superalgebra, with $\tau\indices{_A}$ and $c\indices{^A}$ being related through statistics-reversal and duality --- i.e. $c\indices{^A} (\tau\indices{_B}) = \delta\indices{^A_B}$ and $\deg(c^A) = \deg(\tau_A) + 1$. We define the \emph{BRST operator} on $\mathrm{Pol}^\bullet(\Pi\mathfrak{g})$ as the odd nilpotent supervector field
\begin{equation*}
    Q \coloneqq
        \frac{1}{2} (-1)^{\bar{B}} \, f\indices*{^A_C_B} \, c\indices{^B} c\indices{^C} \frac{\partial}{\partial c\indices{^A}} ,
    \qquad\text{where}\ \bar{B} \coloneqq \deg(\tau\indices{_B}) .
\end{equation*}
The sign $(-1)^{\bar{B}}$ can be understood as a Koszul sign rule; it is there to make sense of a commutator of a boson and a fermion (e.g. $[x,u] = -[u,x]$) being flipped to a ghost monomial ($a \gamma = \gamma a$). For ungraded Lie algebras, $\bar{B} = 0$ and we recover the usual BRST operator.

It is known \cite{AKSZ_1995, Ferreira_2025} that $\bigl( \mathrm{CE}^\bullet(\mathfrak{g}) , d_{\text{CE}} \bigr) \cong \bigl( \mathrm{Pol}^\bullet(\Pi\mathfrak{g}) , Q \bigr)$ as differential graded superalgebras, hence their cohomologies are isomorphic --- but one generally refers to ``BRST cohomology of $\mathfrak{g}$'' when dealing with ghosts.

\bigskip

The Chevalley--Eilenberg formulation can be extended to compute cohomologies in any representation $\rho$ of a Lie (super)algebra $\mathfrak{g}$ on a $\mathfrak{g}$-(super)module $V$ by taking the space of cochains
\begin{equation*}
    \mathrm{CE}_\rho^\bullet (\mathfrak{g},V)
        \coloneqq \mathrm{Hom}\bigl( \Lambda^\bullet(\mathfrak{g}) , V \bigr)
\end{equation*}
equipped with the differential
\begin{alignat*}{3}
    d_\mathrm{CE}^\rho
        \colon \Lambda^k(\mathfrak{g}^\ast) \otimes V &\longrightarrow \Lambda^{k+1}(\mathfrak{g}^\ast) \otimes V
    \\
        \omega \otimes v &\longmapsto d_{\mathrm{CE}}(\omega) \otimes v + (-1)^k \, \omega \wedge \theta_v ,
\end{alignat*}
where $\theta_v \colon \mathfrak{g} \longrightarrow V$ is the 1-cochain defined by 
\begin{equation*}
    \theta_v(x) \coloneqq \rho(x) \cdot v ,
        \qquad\forall x\in\mathfrak{g} .
\end{equation*}

In BRST terms, there is an isomorphism $\mathrm{CE}_\rho^\bullet (\mathfrak{g}) \cong \mathrm{Pol}^\bullet(\Pi\mathfrak{g}) \otimes V$, where the latter space is equipped with the odd supervector field
\begin{equation} \label{2 BRST operator V}
    Q_\rho \coloneqq Q \otimes \mathrm{id} + \Theta ,
    \qquad\text{where}\
    \Theta \colon (\phi \otimes v) \longmapsto (-1)^{\operatorname{deg}(\phi)} \, \phi \wedge \theta_v ,
\end{equation}
and $\deg(\phi)$ is the statistics of $\phi$ in $\operatorname{Pol}^\bullet(\Pi\mathfrak{g})$.

\begin{proposition}
The BRST operator $Q_\rho$ is nilpotent.
\end{proposition}
\begin{proof}
$Q$ itself is nilpotent, so
\begin{equation*}
    Q_\rho Q_\rho
        = Q\Theta + \Theta Q + \Theta^2 .
\end{equation*}
In coordinates, $\theta_v$ can be written as the \nth{1} order ghost polynomial $c\indices{^A} \otimes \rho(\tau\indices{_A}) \cdot v$, so that
\begin{align*}
    \Theta^2 (\phi\otimes v)
        &= (-1)^{\deg(\phi)} \, \Theta\bigl( \phi \, c\indices{^A} \otimes \rho(\tau\indices{_A}) \cdot v \bigr)
        \\
        &= (-1)^{\deg(\phi)} (-1)^{\deg(\phi c^A)} \, \phi \, c\indices{^A} c\indices{^B} \otimes \rho(\tau\indices{_B}) \, \rho(\tau\indices{_A}) \cdot v . 
\end{align*}
Since $\deg(\phi \, c\indices{^A}) = \deg(\phi) + \bar{A} + 1$, and\footnote{Note that the Koszul sign $(-1)^{\bar{A}}$ is necessary to make this equality hold.}
\begin{equation*}
    (-1)^{\bar{A}} \, c\indices{^A} c\indices{^B} \otimes \rho(\tau\indices{_B}) \, \rho(\tau\indices{_A})
        = \frac{1}{2} (-1)^{\bar{A}} \, c\indices{^A} c\indices{^B} \otimes \rho\bigl([\tau\indices{_B},\tau\indices{_A}]\bigr) ,
\end{equation*}
it follows that
\begin{align*}
    \Theta^2 (\phi\otimes v)
        &= -\frac{1}{2} (-1)^{\bar{A}} \, \phi \, c\indices{^A} c\indices{^B} \otimes \rho\bigl([\tau\indices{_B},\tau\indices{_A}]\bigr) \cdot v . 
\end{align*}
In turn,
\begin{align*}
    \Theta Q(\phi\otimes v)
        &= (-1)^{\deg(Q\phi)} \, Q(\phi) \wedge \theta_v ,
    \\
    Q\Theta(\phi\otimes v)
        &= (-1)^{\deg(\phi)} \Bigl( Q(\phi) \wedge \theta_v + (-1)^{\deg(\phi)} \, \phi \wedge Q(\theta_v) \Bigr) ,
\end{align*}
hence $\bigl( \Theta Q + Q\Theta \bigr)(\phi\otimes v) = \phi \wedge Q(\theta_v)$. Back to coordinates,
\begin{align*}
    Q(\theta_v)
        &= Q\bigl( c\indices{^C} \otimes \rho(\tau\indices{_C}) \cdot v \bigr)
        \\
        &= \frac{1}{2} (-1)^{\bar{A}} f\indices*{^C_B_A} \, c\indices{^A} c\indices{^B} \otimes \rho(\tau\indices{_C}) \cdot v
        \\
        &= \frac{1}{2} (-1)^{\bar{A}} \, c\indices{^A} c\indices{^B} \otimes \rho\bigl( [\tau\indices{_B},\tau\indices{_A}] \bigr) \cdot v ,
\end{align*}
therefore $\phi \wedge Q(\theta_v) = -\Theta^2(\phi\otimes v)$, and $Q_\rho Q_\rho = 0$.
\end{proof}

\bigskip

It shall be useful later to have the expression of $\Theta$ acting on a general basis element of $\mathfrak{g}$. Let us set the following convention:
\begin{subequations} \label{2 index convention}
\begin{align}
    i,j,k,\ell
        &\colon \text{indices for $V_0 = \mathbb{K}^m$} ,
    \\
    I,J,K,L
        &\colon \text{indices for $V_1 = \mathbb{K}^n$} ,
\end{align}
so that
\begin{alignat}{3}
    a &= a\indices{^i_j} \otimes \tau\indices{_i^j} ,
    \qquad&\qquad
    \gamma^+ &= \gamma\indices{^I_j} \otimes \tau\indices{_I^j} ,
    \\
    b &= b\indices{^I_J} \otimes \tau\indices{_I^J} ,
    \qquad&\qquad
    \gamma^- &= \gamma\indices{^i_J} \otimes \tau\indices{_i^J} .
\end{alignat}
\end{subequations}
The chosen basis $\{ \tau\indices{_\bullet^\bullet} \}$ is the canonical one in each homogeneous subspace of $\mathfrak{g}$.

\paragraph{Trivial representation}

We know \cite{Ferreira_2025,Fuks_1986} that the BRST operator of $\mathfrak{gl}(m|n)$ acts on fermionic and bosonic ghost coordinates as
\begin{subequations} \label{2 BRST gl with indices}
\begin{alignat}{3}
    Q( a\indices{^i_j} ) 
        &= - a\indices{^i_\ell} \, a\indices{^\ell_j} 
            - \gamma\indices{^i_L} \, \gamma\indices{^L_j} ,
    &\quad
    Q( \gamma\indices{^i_J} )
        &= - a\indices{^i_\ell} \, \gamma\indices{^\ell_J}
            + \gamma\indices{^i_L} \, b\indices{^L_J} ,
    \\
    Q( b\indices{^I_J} )
        &= - b\indices{^I_L} \, b\indices{^L_J}
            - \gamma\indices{^I_\ell} \, \gamma\indices{^\ell_J} ,
    &\quad
    Q( \gamma\indices{^I_j} )
        &= - b\indices{^I_L} \, \gamma\indices{^L_j}
            + \gamma\indices{^I_\ell} \, a\indices{^\ell_j} ,
    \\
\intertext{%
or yet, in matrix form,%
}
    Q(a) &= - aa - \gamma^- \gamma^+ ,
    \qquad&\qquad
    Q(\gamma^-) &= -a\gamma^- + \gamma^- b ,
    \nonumber \\
    Q(b) &= - bb - \gamma^+ \gamma^- ,
    \qquad&\qquad
    Q(\gamma^+) &= -b\gamma^+ + \gamma^+ a .
    \nonumber 
\end{alignat}
\end{subequations}
We shall prefer the matrix form most of the time. The expressions above will be useful in \Cref{subsec: H gl trivial} when reviewing the cohomology of $\mathfrak{gl}(m|n)$.

As already stated, our main interest is in the subspace $\mathfrak{psl}(m|m)$ of $\mathfrak{gl}(n|n)$. Its BRST operator, just like its Lie superbracket, is nothing but the ``restriction'' of that of $\mathfrak{gl}(n|n)$.

Formally speaking, if $\pi \colon \mathfrak{gl}(n|n) \longtwoheadrightarrow \mathfrak{psl}(n|n)$ is the map that restricts to the kernel of the supertrace and removes the central ideal (effectively projecting out the subalgebra generated by $\mathbb{1}_{2n}$ and $\mathbb{J}$), then
\begin{equation} \label{2 Q psl vs Q gl}
    Q^{\mathfrak{psl}} \, \pi(C)
        = (\pi\wedge\pi) \, Q^\mathfrak{gl}(C) ,
        \qquad\forall C\in \operatorname{Pol}^1\bigl( \Pi\mathfrak{gl}(n|n) \bigr) ,
\end{equation}
extending it to all ghost polynomials by derivation.

Having \eqref{2 Q psl vs Q gl} in mind, we shall use ``$Q$'' for the BRST operator of $\mathfrak{gl}(n|n)$, introducing the necessary projections when needed, and ``$Q^\mathfrak{psl}$'' for the operator of $\mathfrak{psl}(n|n)$.

\paragraph{Adjoint representation}

When $\rho$ is the adjoint representation, we have
\begin{align*}
    Q_{\operatorname{ad}} \bigl( c\indices{^A} \otimes \tau\indices{_B} \bigr)
        = \frac{1}{2}(-1)^{\bar{C}} f\indices*{^A_D_C} \, c\indices{^C} c\indices{^D} \otimes \tau\indices{_B}
            - (-1)^{\bar{A}} f\indices*{_C_B^D} \, c\indices{^A} c\indices{^C} \otimes  \tau\indices{_D} .
\end{align*}
Writing a general $\mathfrak{sl}_m$-valued polynomial as $\phi = \phi\indices{^i_j} \otimes \tau\indices{_i^j}$, where $\phi\indices{^i_j}$ is a scalar-valued polynomial in ghost coordinates, we have
\begin{subequations} \label{2 Theta on coordinates}
\begin{align}
    \Theta\bigl( \phi\indices{^i_j} \otimes \tau\indices{_i^j} \bigr)
        &= \bigl( a\phi - (-1)^{\deg(\phi)} \, \phi a \bigr)
            + (-1)^{\deg(\phi)} \bigl( \gamma^+ \phi - \phi \gamma^- \bigr) ,
\\
\intertext{%
where, again, $\deg(\phi)$ is the degree of $\phi$ in $\operatorname{Pol}(\Pi\mathfrak{g})$, and $\phi\indices{^k_j}$ is understood as the entries of a matrix. Similarly,%
}
    \Theta\bigl( \phi\indices{^I_{\!J}} \otimes \tau\indices{_I^J} \bigr)
        &= \bigl( b \phi - (-1)^{\deg(\phi)} \, \phi b \bigr)
            + (-1)^{\deg(\phi)} \bigl( \gamma^- \phi - \phi\gamma^+ \bigr) ,
    \\
    \Theta\bigl( \phi\indices{^i_J} \otimes \tau\indices{_i^J} \bigr)
        &= \bigl( a\phi - (-1)^{\deg(\phi)} \, \phi b \bigr)
            + (-1)^{\deg(\phi)} \bigl( \gamma^+ \phi + \phi \gamma^+ \bigr) ,
    \\
    \Theta\bigl( \phi\indices{^I_j} \otimes \tau\indices{_I^j} \bigr)
        &= \bigl( b \phi - (-1)^{\deg(\phi)} \, \phi a \bigr)
            + (-1)^{\deg(\phi)} \bigl( \gamma^- \phi + \phi \gamma^- \bigr) .
\end{align}
\end{subequations}

\subsection{The Hochschild--Serre spectral sequence}
\label{subsec: spectralseq}

The main tool to compute cohomology in this paper is the Hochschild--Serre (HS) spectral sequence \cite{Hochschild-Serre_1953}. We start by introducing the basic notions regarding spectral sequences and we refer the reader to references on algebraic topology for more details, e.g. \cite{Weibel_1994, Hatcher_2004}.

\begin{definition}
A \emph{cohomological spectral sequence} $(E,d)$ consists of
\begin{enumerate}
    \item a family $\{ E_r^{p,q} \}$ of abelian groups, for $p$, $q$ and $r\geqslant0$ integers --- and the spaces $E_r^{\bullet,\bullet}$ are said to form the \emph{$r^{\text{th}}$ page} of $E$;
    
    \item homomorphisms $d_r^{p,q} \colon E_r^{p,q} \longrightarrow E_r^{p+r,q-r+1}$, called \emph{differentials}, such that the composition $d_r \circ d_r$ (for the correct upper indices) is zero;
    
    \item and isomorphisms
    \begin{equation} \label{2 def spectral sequence pages}
        E_{r+1}^{p,q} \cong \frac{\ker d_r^{p,q}}{\operatorname{im} d_r^{p-r,q+r-1}} .
    \end{equation}
\end{enumerate}
We say the spectral sequence \emph{abuts to} $E_\infty$ if, for every $p$ and $q$, there exists $r_0 = r_0(p,q)$ such that $E_r^{p,q} \cong E_{r_0}^{p,q}$ for every $r \geqslant r_0$; we then say that $E_\infty^{p,q} \coloneqq E_{r_0}^{p,q}$ is the \emph{limiting term} at $(p, q)$, and $d_{r_0}$ is the \emph{transgression} at $(p,q)$. 
\par\noindent
We say that $E$ \emph{degenerates at $r_0$} if, for all $r\geqslant r_0$, all differentials $d_r$ are trivial.
\end{definition}

Rather than having the entire spectral sequence, we can start at its $r^{\text{th}}$ page $E_r^{p,q}$ and construct the next pages by imposing \eqref{2 def spectral sequence pages} to be equalities. We are then interested at its behaviour as $r\to\infty$. For ungraded Lie algebras, the HS spectral sequence presented below has finitely many nontrivial spaces in its \nth{1} page, so it clearly degenerates at some point.

\begin{theorem}[Hochschild--Serre \cite{Hochschild-Serre_1953}] \label{thm: Hochschild Serre}
Let $\mathfrak{g}$ be a Lie (super)algebra, $\mathfrak{h}$ a Lie sub(super)-algebra of $\mathfrak{g}$, and $V$ a $\mathfrak{g}$-(super)module. There exists a cohomological spectral sequence $(E,d)$ whose \nth{1} page is
\begin{equation}
    E_1^{p,q} \cong H^q \Bigl( \mathfrak{g} ; \mathrm{Hom}\bigl( \Lambda^p\, \mathfrak{h} , V \bigr) \Bigr)
\end{equation}
converging to the cohomology of $\mathfrak{g}$ with coefficients in $V$, i.e.
\begin{equation*}
    H^k(\mathfrak{g};V)
        = \bigoplus_{\substack{p,q=0\\p+q=k}}^k E_\infty^{p,q} .
\end{equation*}
\end{theorem}

For a Lie superalgebra $\mathfrak{g} = \mathfrak{g}_0 \oplus \mathfrak{g}_1$, there is a canonical choice of subsuperalgebra: its bosonic part $\mathfrak{g}_0$. The associated HS spectral sequence (i.e. that of \Cref{thm: Hochschild Serre}) with trivial coefficients has
\begin{equation}
    E_1^{p,q}
        = H^q \bigl( \mathfrak{g}_0 ; \mathrm{S}^p (\mathfrak{g}_1{}^\ast) \bigr) ,
\end{equation}
where $\mathrm{S}^\bullet (\mathfrak{g}_1{}^\ast) \equiv \mathbb{K}[\mathfrak{g}_1]$ is the ring of polynomials on the coordinates of $\mathfrak{g}_1$.

When $\mathfrak{g}_0$ is a semisimple Lie algebra, the \nth{1} page decomposes as
\begin{align}
    E_1^{p,q}
        &\cong H^q (\mathfrak{g}_0) \otimes \operatorname{Inv}_{\mathfrak{g}_0} \mathrm{S}^p ( \mathfrak{g}_1{}^\ast) ,
        \label{2 1st page decomposition}
\end{align}
i.e. as the tensor product of the cohomology of $\mathfrak{g}_0$ and the $\mathfrak{g}_0$-invariant polynomials on the coordinates of $\mathfrak{g}_1$ \cite{Fuks_1986, Cartier_1955}.

\bigskip

The HS spectral sequence is a particular case of a spectral sequence of a filtered complex. This notion will be useful when moving to the next pages of the HS spectral sequence.

\begin{definition}
A \emph{decreasing filtration} on a differential graded module $(\mathcal{C},Q)$ is a family of subspaces $\{ F^p \mathcal{C} \subseteq \mathcal{C} \}_p$ such that $F^{p+1} \mathcal{C} \subseteq F^p \mathcal{C}$ and $Q(F^p \mathcal{C}) \subseteq F^p \mathcal{C}$ for every $p\in\mathbb{Z}$ and $\mathcal{C} = \bigcup_p F^p \mathcal{C}$.
\end{definition}

A decreasing filtration induces short exact sequences
\begin{equation} \label{2 short exact sequence filtration}
  \begin{tikzcd}
    0 \arrow{r}
        & F^{p+1} \mathcal{C} \arrow[hook]{r}{i \vphantom{j}}
        & F^p \mathcal{C} \arrow[two heads]{r}{j}
        & E_0^{p,\bullet} 
        \arrow{r}
        & 0 ,
  \end{tikzcd}
\end{equation}
where
\begin{equation*}
    E_0^{p,q} \coloneqq \frac{F^p \mathcal{C}^{p+q}}{F^{p+1} \mathcal{C}^{p+q}} .
\end{equation*}
is equipped with a differential $d_0$: the unique map defined by $Q$ and $j$ \cite{Hatcher_2004}.

Higher pages of that spectral sequence are defined as the cohomology of the previous ones, and the differentials are schematically given by
\begin{equation} \label{2 dr schematic}
    d_r
        \simeq [\cdot] \circ \bigl( j^\ast \, \underbrace{(i^\ast)^{-1} \cdots (i^\ast)^{-1}}_{r-1\ \text{times}} \, \Delta \bigr) \circ [\cdot]^{-1} ,
\end{equation}
where $[\cdot]$ and $[\cdot]^{-1}$ represent taking the cohomology class (with respect to $d_{r-1}$) and one of its representatives, respectively; $i^\ast$ and $j^\ast$ are the induced morphisms in cohomology from \eqref{2 short exact sequence filtration}. The map $\Delta$ is the connecting homomorphism associated to the long exact sequence in cohomology for \eqref{2 short exact sequence filtration} and is schematically given by
\begin{equation} \label{2 delta schematic}
    \Delta
        \simeq [\cdot] \circ \bigl( i^{-1} \, Q \, j^{-1} \bigr) \circ [\cdot]^{-1} .
\end{equation}
Although schematic, \eqref{2 dr schematic} and \eqref{2 delta schematic} are well-defined maps \cite{Lima_2021}.

\bigskip

The HS spectral sequence of $(\mathfrak{g} , \mathfrak{h} ; V)$ is then the spectral sequence associated to the differential graded module $\mathcal{C}^\bullet \coloneqq \operatorname{Pol}^\bullet_\rho(\Pi\mathfrak{g}) \otimes V$ and the filtration
\begin{equation} \label{2 filtration}
    F^p \mathcal{C}^{p+q}
        \coloneqq \bigl\{ \phi \in \mathcal{C}^{p+q} \,\big|\, \phi(X_1,\ldots,X_{p+q}) = 0, \quad \forall X_1,\ldots,X_{q+1}\in\mathfrak{h} \bigr\} ,
\end{equation}
i.e. $F^p \mathcal{C}^{p+q}$ contains the ghost polynomials of order $p+q$ with \emph{at least} $p$ ghost coordinates on $\mathfrak{g}/\mathfrak{h}$. The expression \eqref{2 delta schematic} for the connecting homomorphism makes it clear how the BRST cohomology is encoded in the HS spectral sequence.

\subsection{The cohomology of \texorpdfstring{$\mathfrak{gl}(m|n)$}{gl(m|n)} with trivial coefficients} \label{subsec: H gl trivial}

In this Subsection, we want to delineate the main arguments in the computation of $H\bigl( \mathfrak{gl}(m|n) \bigr)$ following what was done originally by Fuks in \cite{Fuks_1986}. The statement is as follows.

\begin{theorem}[Fuks \cite{Fuks_1986}] \label{thm: Fuks}
The natural inclusion $\mathfrak{gl}(m|n)_0 \longhookrightarrow \mathfrak{gl}(m|n)$ induces the isomorphism in cohomology
\begin{equation} \label{2 cohomology gl mn}
    H^\bullet \bigl( \mathfrak{gl}(m|n) \bigr)
        \cong H^\bullet \bigl( \mathfrak{gl}_{\max\{m,n\}} \bigr) .
\end{equation}
\end{theorem}

Firstly, recall that the cohomology ring of $\mathfrak{gl}_m$ is the exterior algebra in $m$ generators of degrees $1,3,5,\ldots,2m-1$. Such generators can be seen as the maps $x \longmapsto \operatorname{tr}(x^{2k+1})$ for $x\in\mathfrak{gl}_m$, which in BRST language \eqref{2 gl ghost decomposition} read as
\begin{equation*}
    \varphi_{2k+1} = \operatorname{tr}(a^{2k+1}) .
\end{equation*}

Although $\mathfrak{gl}(m|n)_0$ is not semisimple, the decomposition \eqref{2 1st page decomposition} still holds for $\mathfrak{gl}_m \oplus \mathfrak{gl}_n$, and
\begin{equation} \label{2 ring invariants gl}
    \operatorname{Inv}_{\mathfrak{gl}(m|n)_0} \mathrm{S}^\bullet \bigl( \mathfrak{gl}(m|n)_1{}^\ast \bigr)
        = \mathbb{K}\bigl[ \eta_{2k} \,\big|\, 1\leqslant k \leqslant \min\{m,n\}  \bigr] ,
\end{equation}
where
\begin{equation*}
    \eta_{2k} \coloneqq \operatorname{tr}(\gamma^+\gamma^-)^k
        \colon u \longmapsto \operatorname{tr}\bigl( (u^+ u^-)^k \bigr) .
\end{equation*}
In words: the $\mathfrak{gl}(m|n)_0$-invariant polynomials in the coordinates of $\mathfrak{gl}(m|n)_1$ are algebraic combinations of $\eta_{2k}$'s.

It is worth noting that we only have $\min\{m,n\}$ algebraically independent such maps, because every $\operatorname{tr}(\gamma^+ \gamma^-)^{k}$ for $k > \min\{m,n\}$ can be written in terms of the ones in \eqref{2 ring invariants gl} due to the Cayley--Hamilton theorem.

\bigskip

Thus the \nth{1} page of the HS spectral sequence associated with the pair $\bigl( \mathfrak{gl}(m|n) , \mathfrak{gl}(m|n)_0 \bigr)$ in trivial coefficients is the tensor algebra generated by the fermionic ghost polynomials\footnote{The lower index in each of the generators represents its order as a polynomial.}
\begin{subequations} \label{2 fermionic generators}
\begin{alignat}{5}
    \varphi'_1 &\in E_1^{0,1} ,
    \quad &
    \varphi'_{3} &\in E_1^{0,3} ,
    \quad &
    &\ldots ,
    \quad &
    \varphi'_{2m-1} &\in E_1^{0,2m-1}
    \qquad &&\mathbin{\reflectbox{$\leadsto$}}
    H^\bullet(\mathfrak{gl}_m) ,
    \\
    \varphi''_1 &\in E_1^{0,1} ,
    \quad &
    \varphi''_{3} &\in E_1^{0,3} ,
    \quad &
    &\ldots ,
    \quad &
    \varphi''_{2n-1} &\in E_1^{0,2n-1}
    \qquad &&\mathbin{\reflectbox{$\leadsto$}}
    H^\bullet(\mathfrak{gl}_n) ,
\end{alignat}
\end{subequations} 
and the bosonic ghost polynomials
\begin{align*}
    \eta_{2} \in E_1^{2,0} ,
    \quad
    \eta_{4} \in E_1^{4,0} ,
    \quad
    \ldots,
    \quad
    \eta_{2\min\{m,n\}} \in E_1^{2\min\{m,n\},0} .
\end{align*}

\subsubsection{Back to the \texorpdfstring{\nth{0}}{0th} page}

Although we already know what the \nth{1} page of the HS spectral sequence of $\bigl( \mathfrak{gl}(m|n) , \mathfrak{gl}(m|n)_0 ; \mathbb{K} \bigr)$ looks like, it is a good exercise to understand intuitively what happened in $E_0$ to reach $E_1$

Recall that the differential graded module used to build the HS spectral sequence is the Chevalley--Eilenberg complex itself, viz. $\mathcal{C}^\bullet \cong \operatorname{Pol}^\bullet(\Pi\mathfrak{g})$, with the filtration \eqref{2 filtration}. The space $E_0^{p,q}$ on the \nth{0} page contains the $(p+q)^{\text{th}}$ order polynomials which have exactly $p$ coordinates on $U \equiv \mathfrak{g}_1$.

The differential $d_0$ on the \nth{0} page is the unique map such that $d_0 j = j Q$ for the projection map $j$ in \eqref{2 short exact sequence filtration} and the BRST operator $Q$. Due to the projections, $d_0$ acts on fermionic ghost monomials in the same way as $Q$, but ``forgets'' about the $\gamma\gamma$-term of $Qa$ and $Qb$ in \eqref{2 BRST gl with indices}:
\begin{align*}
    d_0 (a) &= -aa ,
    \\
    d_0 (b) &= -bb .
\end{align*}
Thus the restriction of $d_0$ to $\mathfrak{g}_0$ is the BRST operator of the Lie algebra $\mathfrak{g}_0$, taking its cohomology $H^\bullet(\mathfrak{g}_0)$ to the \nth{1} page of the HS spectral sequence.

On the other hand, the action of $d_0$ on bosonic ghost monomials is exactly that of $Q$, which is the Lie algebra action of $\mathfrak{g}_0$ on $U \equiv \mathfrak{g}_1$. Thus, being ``closed'' is the same as being an invariant polynomial --- yielding the ring \eqref{2 ring invariants gl}.

\subsubsection{Transgressions}

Finally, it can be shown that
\begin{equation*}
    d_{2k}^{} ( \varphi'_{2k-1} )
        = d_{2k}^{} (\varphi''_{2k-1})
        = \eta_{2k}^{}
\end{equation*}
is a transgression map at the spot $(0,2k-1)$ of the HS spectral sequence. \textcite{Fuks_1986} shows that through Weyl algebras.

An intuitive approach is the following. The connecting homomorphism $\Delta$ from the short exact sequence \eqref{2 short exact sequence filtration} acts on the generator $\varphi'_{2k-1} = \operatorname{tr}(a^{2k-1})$ as
\begin{align*}
    \delta\bigl[\operatorname{tr}(a^{2k-1})\bigr]
        &= \bigl[ \bigl( i^{-1} Q \, j^{-1} \bigr) \operatorname{tr}(a^{2k-1}) \bigr]
        \\
        &= \bigl[ \operatorname{tr}(\gamma^+ a^{2k} \gamma^-) \bigr]
        & \in H^{2k}( F^1 \mathcal{C} ) .
\end{align*}
By applying $2k-1$ times the ``map'' $(i^\ast)^{-1}$, the above object must land in $H^{2k}(F^{2k} \mathcal{C})$, which will then by mapped by $j^\ast$ to $E_{2k}^{2k,0}$, finishing the transgression. The only valid choice in $H^{2k}(F^{2k} \mathcal{C})$ is the class corresponding to $\eta_{2k} \equiv \operatorname{tr}(\gamma^+\gamma^-)^k$ --- because $\eta_{2k}$ and $\operatorname{tr}(\gamma^+ a^{2k} \gamma^-)$ were cohomologous before the filtration, and $E_{2k}^{2k,0}$ is 1-dimensional.

Consequently, the fermionic ghost polynomial
\begin{equation*}
    \psi_{2k-1}^{}
        \coloneqq \varphi'_{2k-1} - \varphi''_{2k-1}
        \quad \in \ker d_{2k}^{}
\end{equation*}
survives transgression and goes to $E^{0,2k-1}_\infty$, generating cohomology for $\mathfrak{gl}(m|n)$. That ghost polynomial can be seen as the map $X_0 \longmapsto \operatorname{str}(X_0^{2k-1})$, where ${ X_0 \in \mathfrak{gl}(m|n)_0 }$ is in the standard decomposition \eqref{2 gl decomposition}, i.e. $\psi_{2k-1}$ is the ghost polynomial
\begin{equation*}
    \psi_{2k-1}
        \equiv \operatorname{str}\bigl(c_0^{2k-1}\bigr)
        = \operatorname{tr}(a^{2k-1}) - \operatorname{tr}(b^{2k-1}) .
\end{equation*}

While $k \leqslant \min\{m,n\}$, each pair $\varphi'_{2k-1},\varphi''_{2k-1}$ of fermionic ghost polynomials transgresses to the bosonic ghost polynomials $\eta_{2k}^{}$, and their difference generates coholomogy in the $\infty^{\text{th}}$ page. When $k > \min\{m,n\}$, there are no bosonic ghost generators left, and the remaining $|m-n|$ fermionic ghost generators transgress to zero --- also surviving until $E_\infty$. \Cref{fig: HS gl mn} shows schematically what happens to the generators of $E_1$.

Therefore, there are exactly $\max\{m,n\}$ fermionic ghost generators for the cohomology of $\mathfrak{gl}(m|n)$, and they are in degrees $1,3,\ldots,2\max\{m,n\}-1$, proving \Cref{thm: Fuks}.

\bigskip

\begin{figure}[H]
\centering
\begin{tikzpicture}
\begin{axis}[
    axis lines = middle,
    axis line style = {->, >=Stealth},
    xlabel = {$p$},
    ylabel = {$q$},
    x label style = {anchor = west},
    y label style = {anchor = south},
    xmin = -0.5,    xmax = 12,
    ymin = -0.5,    ymax = 11,
    xtick = {2, 4, 8, 10},
    xticklabels = {$2$, $4$, $2\tilde{n}$, $2\tilde{n}+2$},
    ytick = {1, 3, 7, 9},
    yticklabels = {$1$, $3$, $2\tilde{n}-1$, $2\tilde{n}+1$},
    grid = none,
    axis x line = bottom,
    axis y line = left,
    x axis line style = {draw=none},
    y axis line style = {draw=none},
    label style = {font=\small},
]


\draw (axis cs:-1,0) -- (axis cs:5.9,0);
\draw[->, >=Stealth] (axis cs:6.1,0) -- (axis cs:12,0);
\draw[] (axis cs:5.9,-0.2) -- (axis cs:5.9,+0.2);
\draw[] (axis cs:6.1,-0.2) -- (axis cs:6.1,+0.2);

\draw (axis cs:0,-1) -- (axis cs:0,4.9);
\draw[->, >=Stealth] (axis cs:0,5.1) -- (axis cs:0,11);
\draw[] (axis cs:-0.2,4.9) -- (axis cs:+0.2,4.9);
\draw[] (axis cs:+0.2,5.1) -- (axis cs:-0.2,5.1);

\filldraw[black] (0,1.1) circle (1.5pt);
    \draw (0,1.4) node[anchor=west] {\scriptsize$\varphi'_1$};
\filldraw[black] (0,0.9) circle (1.5pt);
    \draw (0,0.6) node[anchor=west] {\scriptsize$\varphi''_1$};

\filldraw[black] (0,3.1) circle (1.5pt);
    \draw (0,3.4) node[anchor=west] {\scriptsize$\varphi'_3$};
\filldraw[black] (0,2.9) circle (1.5pt);
    \draw (0,2.6) node[anchor=west] {\scriptsize$\varphi''_3$};

\filldraw[black] (0,7.1) circle (1.5pt);
    \draw (0,7.4) node[anchor=west] {\scriptsize$\varphi'_{2\tilde{n}-1}$};
\filldraw[black] (0,6.9) circle (1.5pt);
    \draw (0,6.6) node[anchor=west] {\scriptsize$\varphi''_{2\tilde{n}-1}$};

\filldraw[black] (0,9) circle (1.5pt);
    \draw (0,9.4) node[anchor=west] {\scriptsize$\varphi'_{2\tilde{n}+1}$};

\filldraw[black] (2,0) circle (1.5pt);
    \draw (2.25,0.4) node {\footnotesize$\eta_2$};

\filldraw[black] (4,0) circle (1.5pt);
    \draw (4.25,0.4) node {\footnotesize$\eta_4$};

\filldraw[black] (8,0) circle (1.5pt);
    \draw (8.25,0.4) node {\footnotesize$\eta_{2\tilde{n}}$};

\draw[black, fill=white] (10,0) circle (1.5pt);

\draw[|->, red, shorten <=3pt, shorten >=2pt] 
    (axis cs:0,1) -- (axis cs:2,0);
\draw[|->, red, shorten <=3pt, shorten >=2pt]
    (axis cs:0,3) -- (axis cs:4,0);
\draw[|->, red, shorten <=3pt, shorten >=2pt]
    (axis cs:0,7) -- (axis cs:8,0);
\draw[|->, red, shorten <=3pt, shorten >=2pt]
    (axis cs:0,9) -- (axis cs:10,0);
\draw[white] (axis cs:0,11) --node[midway, red]{$\sdots$} (axis cs:12,0);

\end{axis}
\end{tikzpicture}
\caption{\label{fig: HS gl mn}
    Schematization of \nth{1} page and higher differentials of the HS spectral sequence for $\mathfrak{gl}(m|n)$, with $\tilde{n} \coloneqq \min\{m,n\}$. Fully black dots correspond to the ghost generators of $E_1^{\bullet,\bullet}$ as a tensor algebra. While $k \leqslant \tilde{n}$, the fermionic ghost generators $\varphi'_{2k-1},\varphi''_{2k-1}$ transgress to $\eta_{2k}$ in pairs, thus $\psi_{2k-1}$ transgresses to zero. When $k > \tilde{n}$, there is only one fermionic ghost generator $\varphi'_{2k-1}$ and no bosonic ghost generators left for it to transgress to.}
\end{figure}
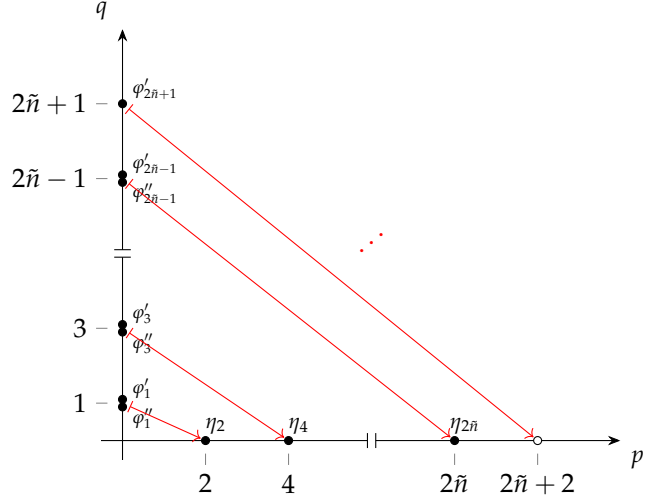

\subsubsection{Invariants and covariants} \label{sec: invariant theory}

Already hinted in \eqref{2 ring invariants gl}, the computation of Lie superalgebra cohomology through its HS spectral sequence is deeply related to invariant theory. For that reason, we conclude this Section of preliminaries by defining the jargon of invariants and covariants as taken from \cite{Weyl_1966, Kraft_1996, Procesi_2007}.

\begin{definition}
Let $V$ and $W$ be $G$-modules. The submodule of \emph{$G$-invariants of $V$} is
\begin{equation} \label{def Lie group invariants}
    V^G \equiv \operatorname{Inv}_G (V)
        \coloneqq \bigl\{ v\in V \,\big|\, g \vartriangleright v = v ,\ \forall g\in G \bigr\} .
\end{equation}
A \emph{covariant\footnote{Also called \emph{intertwiner}.} of $V$ of type $W$} is a polynomial map $\varphi \colon V \longrightarrow W$ that preserves the $G$-actions, viz.
\begin{equation} \label{def equivariant morphism}
    \varphi\bigl( g \vartriangleright v \bigr)
        = g \vartriangleright \varphi(v) , 
    \qquad \forall g\in G, \ \forall v\in V .
\end{equation}
\end{definition}

In the space of $G$-module homomorphisms $\mathrm{Hom}(V,W) \cong V^\ast \otimes W$, there is an induced $G$-action on the tensor product. Firstly,
\begin{align}
    \bigl( g \vartriangleright \phi \bigr) \colon
        V &\longrightarrow \mathbb{K}
        \nonumber \\
        v &\longmapsto \phi\bigl( g^{-1} \vartriangleright v \bigr)
        \label{2 def contragradient action}
 \end{align}
defines the \emph{contragradient action} of $g \in G$ on $\phi \in V^\ast$, so the action of $g\in G$ on $\phi\otimes w \in \operatorname{Hom}(V,W)$ is
\begin{equation*}
    \bigl( g \vartriangleright (\phi \otimes w) \bigr) (v)
        \coloneqq \phi\bigl( g^{-1} \vartriangleright v \bigr) \, \bigl( g \vartriangleright w \bigr) ,
        \qquad \forall v\in V .
\end{equation*}
It is easy to see that $\phi \otimes w$ is a covariant map if and only if it is a $G$-invariant of $\operatorname{Hom}(V,W)$ under the action above.

\bigskip

The HS spectral sequence will require knowledge about Lie algebra invariants and covariants, which are defined analogously: if $\mathfrak{g}$ is a Lie algebra, and $V$ is a $\mathfrak{g}$-module, then the submodule of $\mathfrak{g}$-invariants is
\begin{equation} \label{def Lie algebra invariants}
    V^\mathfrak{g} \equiv \operatorname{Inv}_\mathfrak{g} (V)
        \coloneqq \bigl\{ v\in V \,\big|\, x \vartriangleright v = 0, \ \forall x\in\mathfrak{g} \bigr\} .
\end{equation}

The connection between $G$-invariants and $\mathfrak{g}$-invariants is immediate:

\begin{lemma} \label{lemma: connected group}
Let $G$ be a Lie group, with $\mathfrak{g}$ its Lie algebra, and $V$ a $G$-module. If $G$ is connected, the $G$-invariants of $V$ are in 1-to-1 with the $\mathfrak{g}$-invariants of $V$ (with respect to the infinitesimal action of $G$ on $V$).
\end{lemma}

\section{Cohomology with trivial coefficients} \label{sec: trivial coefficients}

In this Section, we explicitly compute $H\bigl( \mathfrak{psl}(n|n) ; \mathbb{K} \bigr)$ and $H\bigl( \mathfrak{sl}(m|n) ; \mathbb{K} \bigr)$ via their HS spectral sequences.
\begin{itemize}
    \item \Cref{subsec: invariant polynomials} is a review of the Fundamental Theorems of representation theory for $\mathrm{SL}_n$ and how they apply to our problem;

    \item In \Cref{subsec: HS trivial coefficients}, we give the closed expression for $H\bigl( \mathfrak{psl}(n|n) \bigr)$;

    \item \Cref{subsec: examples trivial} shows explicit bases of $H^k \bigl( \mathfrak{psl}(n|n) \bigr)$ for small $n$ and $k$; and
    
    \item Finally, in \Cref{subsec: sl trivial} we give the expression for $H\bigl( \mathfrak{sl}(m|n) \bigr)$ for arbitrary $m,n \in \mathbb{N}$.
\end{itemize}

\bigskip

The bosonic subspaces of $\mathfrak{sl}(m|n)$ and $\mathfrak{psl}(n|n)$ were presented in \eqref{2 sl m n 0} and \eqref{2 psl m n 0}, respectively. They are semisimple, so we can use the decomposition \eqref{2 1st page decomposition}.

The cohomology of $\mathfrak{sl}_n$ is well-known: it is the exterior algebra in $n-1$ generators of degrees $3,5,\ldots,2n-1$. It remains to understand the ring of invariant polynomials
\begin{align} \label{3 Inv complete}
    \operatorname{Inv}_{\mathfrak{sl}(m|n)_0} \mathrm{S}^\bullet ( U^\ast )
        \cong \operatorname{Inv}_{\mathfrak{u}_1} \operatorname{Inv}_{\mathfrak{sl}_m} \operatorname{Inv}_{\mathfrak{sl}_n} \mathrm{S}^\bullet (U^\ast) .
\end{align}
We will, for now, consider only the right-most piece of \eqref{3 Inv complete}, namely the space 
\begin{equation*}
    \operatorname{Inv}_{\mathfrak{sl}_n} \mathrm{S}^\bullet (U^\ast) .
\end{equation*}

\subsection{Invariant polynomials} \label{subsec: invariant polynomials}

The fermionic subspace of $\mathfrak{sl}(m|n)$ is the same as that of $\mathfrak{gl}(m|n)$, presented in \eqref{2 gl m n 1} and \eqref{2 index structure of U}. Seeing $U^+$ and $U^-$ as the spaces of $n\times m$ and $m\times n$ matrices over $\mathbb{K}$, we can regard them as $\mathfrak{sl}_n$-modules by taking
\begin{align*}
    U^+
        &\cong (V_1)^m ,
    \\
    U^-
        &\cong (V_1{}^\ast)^m ,
\end{align*}
where the $\mathfrak{sl}_n$-action on $V_1$ is by left-multiplication (the natural action on $n$-dimensional column vectors), and the $\mathfrak{sl}_n$-action on $V_1{}^\ast$ is by right-multiplication (the contragradient action on $n$-dimensional row vectors, see \eqref{2 def contragradient action}). Hence, we wish to characterize the ring
\begin{equation} \label{3 invariants sl n}
    \operatorname{Inv}_{\mathfrak{sl}_n} \mathrm{S}^\bullet (U^\ast)
        \cong \mathbb{K}\bigl[ (V_1)^m \oplus (V_1{}^\ast)^m \bigr]^{\mathrm{SL}_n} .
\end{equation}

The special linear Lie group $\mathrm{SL}_n$ is connected, thus we can use \Cref{lemma: connected group} and the First Fundamental Theorem (FFT) for $\mathrm{SL}_n$.

\begin{theorem}[FFT for $\mathrm{SL}_n$ \cite{Kraft_1996, Weyl_1966}] \label{thm: fft}
Let $V = \mathbb{K}^n$. The ring $\mathbb{K}[ V^p \oplus (V^\ast)^q ]^{\mathrm{SL}_n}$ of invariant polynomials of $V$ is generated by the contractions 
\begin{equation} \label{3 FFT contractions}
    \langle i \,|\, j \rangle
        \colon (v_1,\ldots,v_p;\xi_1,\ldots,\xi_q)
        \longmapsto \langle v_i | \xi_j \rangle \coloneqq \xi_j (v_i)
\end{equation}
and the minor $n\times n$ determinants 
\begin{subequations} \label{3 FFT minors}
\begin{alignat}{3}
    \label{3 FFT det}
    [i_1 , \ldots, i_n]^\uparrow
        &\colon (v_1,\ldots,v_p;\xi)
        &\longmapsto \det(v_{i_1},\ldots,v_{i_n})
        \qquad &(\text{if}\ p \geqslant n) ,
    \\
    \label{3 FFT det ast}
    [j_1 , \ldots , j_n]^\downarrow
        &\colon (v;\xi_1,\ldots,\xi_q)
        &\longmapsto \det(\xi_{j_1},\ldots,\xi_{j_n})
        \qquad &(\text{if}\ q \geqslant n) ,
\end{alignat}
\end{subequations}
for $i \in \{1,\ldots,p\}$ and $j \in \{ 1,\ldots,q \}$, with $v \in V$ vectors, and $\xi \in V^\ast$ covectors.
\end{theorem}

The contractions \eqref{3 FFT contractions} are precisely those found by \textcite{Fuks_1986}, as they are also the invariants of $\mathrm{GL}_n$. Thus the only ``new'' invariants (with respect to $\mathrm{GL}_n$) are the minor determinants \eqref{3 FFT det} and \eqref{3 FFT det ast}.

In particular, if $p,q<n$, we are left only with the contractions \eqref{3 FFT contractions}. Back to \eqref{3 Inv complete} and \eqref{3 invariants sl n}, we see that $m\neq n$ implies that the only $\mathfrak{sl}(m|n)_0$-invariants are those of $\mathfrak{gl}(m|n)_0$. If $m=n$, the $\mathfrak{u}_1$ sector is generated by the identity, thus acting trivially on $\mathfrak{sl}(m|m)_1$, and the minor determinants become precisely the determinants, viz.
\begin{align*}
    [1,\ldots,n]^\uparrow
        &\equiv \det(\gamma^+)
        \eqqcolon \delta^+ ,
    \\
    [1,\ldots,n]^\downarrow
        &\equiv \det(\gamma^-) 
        \eqqcolon \delta^- .
\end{align*}

\bigskip

There are, however, non-trivial relations between the generating invariant polynomials in the FFT, which are encoded in the Second Fundamental Theorem (SFT) for $\mathrm{SL}_n$.

\begin{theorem}[SFT for $\mathrm{SL}_n$ \cite{Weyl_1966}] \label{thm: sft}
All relations among the generators in \eqref{3 FFT contractions} and \eqref{3 FFT minors} are algebraic consequences of the following:
\begin{subequations}
\begin{alignat}{3}
    \sum_{t=0}^{n} (-1)^t \,[i_0,\ldots,\widehat{i_t},\ldots,i_n]^\uparrow \; \langle i_t | j \rangle 
        &= 0
        &(\text{if}\ p > n) ,
    \label{3 SFT a}
    \\
    \sum_{t=0}^{n} (-1)^t \, [j_0,\ldots,\widehat{j_t},\ldots,j_n]^\downarrow \; \langle i | j_t \rangle
        &= 0
        &(\text{if}\ q > n) ,
    \label{3 SFT b}
    \\
    \sum_{t=0}^{n} (-1)^t \, [i_0,\ldots,\widehat{i_t},\ldots,i_n]^\uparrow \; [i_t,k_2,\ldots,k_n]^\uparrow
        &= 0
        &(\text{if}\ p > n) ,
    \label{3 SFT c}
    \\
    \sum_{t=0}^{n} (-1)^t \, [j_0,\ldots,\widehat{j_t},\ldots,j_n]^\downarrow \; [j_t,\ell_2,\ldots,\ell_n]^\downarrow
        &= 0 
        &(\text{if}\ q > n) ,
    \label{3 SFT d}
    \\
    [i_1,\ldots,i_n]^\uparrow \; [j_1,\ldots,j_n]^\downarrow
    - \begin{vmatrix}
            \langle i_1 | j_1 \rangle & \cdots & \langle i_1 | j_n \rangle
            \\
            \vdots & \ddots & \vdots \\
            \langle i_n | j_1 \rangle & \cdots & \langle i_n | j_n \rangle
        \end{vmatrix} 
        &= 0 
        \quad &(\text{if}\ p,q \geqslant n) ,
        \!\!\!\!\!\!
    \label{3 SFT e}
\end{alignat}
\end{subequations}
where a hat over an index indicates the omission of that index from the sequence.
\end{theorem}

Although not explicit in \cite{Weyl_1966}, the relations \eqref{3 SFT a}--\eqref{3 SFT d} only make sense when $p>n$ (or $q>n$, accordingly). Indeed, taking $p=n$ and recalling that the $n+1$ indices $i_0,\ldots,i_n$ must range in $\{1,\ldots,p=n\}$, at least two of them must be equal. This implies that the relations \eqref{3 SFT a} and \eqref{3 SFT c} cannot hold for $p=n$, otherwise $[1,\ldots,n] \, \langle i | j \rangle$ would vanish for all $i$ and $j$.

\bigskip

Applying \Cref{thm: fft} and \Cref{thm: sft} to our cases of interest, we present following proposition.

\begin{proposition} \label{prop: sl mn invariants}
If $m\neq n$, then
\begin{align*}
    \operatorname{Inv}_{\mathfrak{sl}(m|n)_0} \mathrm{S}^\bullet \bigl( \mathfrak{sl}(m|n)_1{}^\ast \bigr)
        &= \operatorname{Inv}_{\mathfrak{gl}(m|n)_0} \mathrm{S}^\bullet\bigl( \mathfrak{gl}(m|n)_1{}^\ast \bigr)
        \\
        &= \mathbb{K}\bigl[ \eta_{2k} \,\big|\, 1 \leqslant k \leqslant \min\{m,n\} \bigr] .
\end{align*}
If $m = n$, then
\begin{equation} \label{3 def R}
    \operatorname{Inv}_{\mathfrak{sl}(m|m)_0} \mathrm{S}^\bullet\bigl( \mathfrak{sl}(m|m)_1{}^\ast \bigr)
        = \mathbb{K}\bigl[ \eta_{2k} ,\ \delta^+ ,\ \delta^- \,\big|\, 1 \leqslant k \leqslant m \bigr] \big/ \mathcal{J} ,
\end{equation}
where $\mathcal{J}$ is the ideal generated by the relation \eqref{3 SFT e}.
\end{proposition}
\begin{proof}
As already argued, when $m\neq n$, the minor determinants are not invariant polynomials: if $m>n$, the $n\times n$ minors cannot be invariants of $\mathrm{SL}_m$. Thus we have the same invariant polynomials as we had for $\mathfrak{gl}(m|n)$.
\par\noindent
When $m=n$, we have the determinants $\delta^\pm$, and the only surviving relation from \Cref{thm: sft} is \eqref{3 SFT e}.
\end{proof}

For what follows, let $\mathcal{R}^\bullet$ be the ring defined on \eqref{3 def R}, which is naturally graded by the polynomial degree. Since $\mathfrak{sl}(n|n)_0 \cong \mathfrak{psl}(n|n)_0 \oplus \mathfrak{i}_{2n}$ (where $\mathfrak{i}_{2n}$ is the ideal generated by the identity matrix), it is immediate to see that
\begin{equation*}
    \mathcal{R}^\bullet
        = \operatorname{Inv}_{\mathfrak{psl}(n|n)_0} \mathrm{S}^{\bullet}\bigl( \mathfrak{psl}(n|n)_1 \bigr) .
\end{equation*}
In words: the rings of invariant polynomials appearing in the HS spectral sequences of $\mathfrak{sl}(n|n)$ and $\mathfrak{psl}(n|n)$ are the same.

The rings described in \Cref{prop: sl mn invariants} appeared in a different fashion in \cite[Section 8.11]{Boe_2006}, computed through the Luna--Richardson theory. Our computation here is based on the first and second fundamental theorems for invariant theory.

\subsection[HS spectral sequence in trivial coefficients]{Hochschild--Serre spectral sequence with trivial coefficients} \label{subsec: HS trivial coefficients}

The previous discussion culminates in the first main theorem of this paper.

\begin{theorem} \label{thm: psl trivial coefficients}
The cohomology of $\mathfrak{psl}(n|n)$ in trivial coefficients is given by
\begin{equation} \label{3 cohomology psl m m}
    H \bigl( \mathfrak{psl}(n|n) \bigr)
        \cong H(\mathfrak{sl}_n) \otimes \mathbb{K}\bigl[ \eta_2 ,\ \delta^+ ,\ \delta^- \bigr] \big/ \mathcal{J}_0 ,
\end{equation}
where $\mathcal{J}_0$ is the ideal generated by the relation
\begin{equation} \label{3 nontrivial relation cohomology}
    \delta^+ \, \delta^-
        = \frac{1}{n!} \, (\eta_2)^n .
\end{equation}
\end{theorem}
\begin{proof}
From \Cref{thm: Hochschild Serre}, we know that the \nth{1} page of the Hochschild--Serre spectral sequence for $\bigl( \mathfrak{psl}(n|n) , \mathfrak{sl}_n \oplus\mathfrak{sl}_n \bigr)$ with trivial coefficients is
\begin{equation*}
    E_1^{p,q}
        \cong H^q \bigl( \mathfrak{sl}_n \oplus \mathfrak{sl}_n \bigr)
            \otimes \mathcal{R}^p .
\end{equation*}
It is the tensor algebra generated by the fermionic ghost polynomials
\begin{alignat*}{5}
    \varphi'_3 &\in E_1^{0,3} ,
    \quad
    \varphi'_5 &&\in E_1^{0,5} ,
    \quad &&\ldots ,
    \quad
    \varphi'_{2m-1} &&\in E_1^{0,2m-1} ,
    \\
    \varphi''_3 &\in E_1^{0,3} ,
    \quad
    \varphi''_5 &&\in E_1^{0,5} ,
    \quad &&\ldots ,
    \quad
    \varphi''_{2n-1} &&\in E_1^{0,2n-1}
\end{alignat*}
--- each line for one of the copies of $\mathfrak{sl}_n$ ---, and the bosonic ghost polynomials
\begin{align*}
    \eta_2 \in E_1^{2,0} ,
    \quad
    \eta_4 \in E_1^{4,0} ,
    \quad
    \ldots,
    \quad
    \eta_{2n} \in E_1^{2n,0} ,
    \qquad\quad
    \delta^+ , \delta^- \in E_1^{n,0} ,
\end{align*}
which generate the ring $\mathcal{R}$ of invariant polynomials.
\par\noindent
The BRST differential of $\mathfrak{psl}(n|n)$ is inherited from $\mathfrak{gl}(n|n)$ according to \eqref{2 Q psl vs Q gl}, so we already know that $\varphi'_{2k-1} , \varphi''_{2k-1}$ transgress to $\eta_{2k}$. Again, this means that
\begin{equation*}
    \psi_{2k-1}^{}
        \coloneqq \varphi'_{2k-1} - \varphi''_{2k-1}
        = \operatorname{str}\bigl(C_0^{2k-1}\bigr)
\end{equation*}
survives transgression and goes to $E_\infty$. Moreover, these transgressions exhaust all fermionic ghost polynomials and all trace-like bosonic ghost polynomials but one: $\eta_2$.
\par\noindent
The bosonic ghost polynomial $\eta_2$ generates the \nth{2} cohomology class of $\mathfrak{psl}(n|n)$ associated with the central extension $\mathfrak{sl}(n|n)$; we show it explicitly for $\mathfrak{psl}(2|2)$ below, but the argument holds for any $n$. The determinant-like bosonic ghost polynomials $\delta^\pm$ can never be in the image of HS differentials (since we already know how they act on all other generators), but will always be in their kernels (as they lie on the $E^{\bullet,0}$ row), hence going to $E_\infty$ too. Therefore, we already know that the cohomology in trivial coefficients must be generated by the tensor product of
\begin{align*}
    H(\mathfrak{sl}_n) \cong \Lambda \bigl( \mathbb{K}\bigl\{ \psi_3 , \ldots, \psi_{2n-1} \bigr\} \bigr)
\end{align*}
with the subring of polynomials $\mathbb{K}[\eta_2, \ \delta^+,\ \delta^-]$.
\par\noindent
Finally, these generators are not completely free. By the Cayley--Hamilton theorem and Newton's identities, we know that the $2n^{\text{th}}$ order polynomial
\begin{equation*}
    \delta^+ \, \delta^-
        = \det\bigl( \gamma^+ \gamma^- \bigr)
\end{equation*}
can be written in terms of $\eta_{2k} = \operatorname{tr}\bigl(\gamma^+ \gamma^-\bigr)^k$. But, apart from $\eta_2$, all the other $\eta_{2k}$'s are cohomologous to zero by the transgressions, therefore
\begin{equation*}
    \delta^+ \, \delta^-
        = \frac{1}{n!} \, (\eta_2)^n 
        \qquad\text{in cohomology} ,
\end{equation*}
concluding the proof.
\end{proof}

\subsection{Explicit examples} \label{subsec: examples trivial}

In the following Subsubsections, we will look at the lower order cohomology groups for $n \leqslant4$ and give their explicit generators.

\subsubsection{Cohomology of \texorpdfstring{$\mathfrak{psl}(1|1)$}{psl(1|1)}}

The case $n=1$ is pathological in the sense that $\mathfrak{psl}(1|1)$ has \emph{no fermionic ghosts at all} (because $\mathfrak{sl}_1 = 0$). The two bosonic ghosts, $\gamma^\pm$, turn out to be the invariants themselves, so
\begin{equation*}
    H\bigl( \mathfrak{psl}(1|1) \bigr)
        = \mathbb{K}[ \delta^+ , \delta^- ] .
\end{equation*}
Note that this cohomology ring is free.

\subsubsection{Cohomology of \texorpdfstring{$\mathfrak{psl}(2|2)$}{psl(2|2)} and its central extensions}

Moving on to $n=2$, the only non-trivial cohomology of $\mathfrak{sl}_2$ is in \nth{3} order, so the only fermionic ghost generator for $\mathfrak{psl}(2|2)$ is $\psi_3^{} = \varphi'_3 - \varphi''_3$. The determinants $\delta^\pm$ are both of \nth{2} order. \Cref{table: psl 2 2} shows the low order cohomology spaces of $\mathfrak{psl}(2|2)$, their generators from the HS spectral sequence, and their dimensions (Betti numbers).

\begin{table}[H]
\centering
\begin{tabular}{|c|c|c|}
    \hline\hline
    $\boldsymbol{k}$
        & \textbf{generators of} $\boldsymbol{H^k\bigl( \mathfrak{psl}(2|2) \bigr)}$
        & $\boldsymbol{\beta^k}$
    \\
    \hline\hline
    $1$ & $1$ & $1$
    \\ \hline
    $2$ & $\eta_2$, $\delta^\pm$
        & $3$
    \\ \hline
    $3$ & $\psi_3$
        & $1$
    \\ \hline
    $4$ & $(\eta_2)^2$, $(\delta^\pm)^2$, $\eta_2 \, \delta^\pm$
        & $5$
    \\ \hline
    $5$ & $\psi_3 \eta_2$, $\psi_3 \delta^\pm$
        & $3$
    \\ \hline
    \multicolumn{3}{c}{$\vdots$}
\end{tabular}
\caption{Low order cohomology spaces of $\mathfrak{psl}(2|2)$ in trivial coefficients, their generators and their dimensions (Betti numbers).}
\label{table: psl 2 2}
\end{table}


In general, we have the recursion formulas
\begin{align*}
    H^{2k} \bigl( \mathfrak{psl}(2|2) \bigr)
        &\cong \eta_2 \cdot H^{2(k-1)} \bigl( \mathfrak{psl}(2|2) \bigr)
            \;\oplus\; \mathbb{K} \bigl\{ (\delta^\pm)^k \bigr\} ,
    \\
    H^{2k+1} \bigl( \mathfrak{psl}(2|2) \bigr)
        &\cong \psi_3 \cdot H^{2(k-1)} \bigl( \mathfrak{psl}(2|2) \bigr) ,
\end{align*}
which correspond to the Betti numbers ($k>0$)
\begin{align*}
    \beta^{2k}\bigl( \mathfrak{psl}(2|2) \bigr)
        &= 2k+1 ,
    \\
    \beta^{2k+1} \bigl( \mathfrak{psl}(2|2) \bigr)
        &= 2k-1 .
\end{align*}

\bigskip

The \nth{2} cohomology of $\mathfrak{psl}(2|2)$ is 3-dimensional, as expected \cite{Scheunert-Zhang_1998}; hence, there are three inequivalent central extensions. The first one is obvious,
\begin{equation*}
    \eta_2 \ \leftrightsquigarrow \ \mathfrak{sl}(2|2)
        \cong \mathfrak{psl}(2|2) \oplus \mathbb{K} ,
\end{equation*}
with Lie superbracket given by
\begin{align*}
    \bigl[ \tau\indices*{^+_\alpha} , \tau\indices*{^-_\beta} \bigr]_{\mathfrak{sl}}
        &= \bigl[ \tau\indices*{^+_\alpha} , \tau\indices*{^-_\beta} \bigr]_{\mathfrak{psl}} + \delta\indices{_\alpha_\beta} \, \mathbb{1} ,
    \\
    \bigl[ \mathbb{1}^+ , \mathbb{1}^- \bigr]_{\mathfrak{sl}}
        &= \bigl[ \mathbb{1}^+ , \mathbb{1}^- \bigr]_{\mathfrak{psl}} + \mathbb{1} ,
\end{align*}
where $\{ \mathbb{1}^+ \} \cup \{ \tau\indices*{^+_\alpha} \}_\alpha^{}$ is a basis for $U^+ \cong \mathfrak{gl}_2$, with $\mathbb{1}^+$ the identity matrix, and $\{ \tau\indices*{^+_\alpha} \}_\alpha$ is a Pauli basis for $\mathfrak{sl}_2 \subset U^+$. Similarly for $U^- \cong \mathfrak{gl}_2$.

The other two ghosts, $\delta^\pm = \det(\gamma^\pm)$, 
are associated with the central extensions
\begin{align*}
    \delta^\pm \ \leftrightsquigarrow \ \mathfrak{psl}^\pm(2|2)
        \coloneqq \mathfrak{psl}(2|2) \oplus \mathfrak{e}^\pm ,
\end{align*}
whose Lie superbrackets are given by
\begin{align*}
    \bigl[ \tau\indices*{^\pm_\alpha} , \tau\indices*{^\pm_\beta} \bigr]_{\mathfrak{psl}^\pm}
        &= \bigl[ \tau\indices*{^\pm_\alpha} , \tau\indices*{^\pm_\beta} \bigr]_{\mathfrak{psl}} + \delta\indices{_\alpha_\beta} \, e^\pm ,
    \\
    \bigl[ \mathbb{1}^\pm , \mathbb{1}^\pm \bigr]_{\mathfrak{psl}^\pm}
        &= \bigl[ \mathbb{1}^\pm , \mathbb{1}^\pm \bigr]_{\mathfrak{psl}} + e^\pm ,
\end{align*}
where $e^\pm$ is the generator of the abelian Lie algebra $\mathfrak{e}^\pm$.

\subsubsection{Cohomology of \texorpdfstring{$\mathfrak{psl}(3|3)$}{psl(3|3)} and its 2-extensions}

\Cref{thm: psl trivial coefficients} tells us that the cohomology of $\mathfrak{psl}(3|3)$ with trivial coefficients is generated by the fermionic ghost polynomials $\psi_3$ and $\psi_5$ and the bosonic ghost polynomials $\eta_2$ and $\delta^\pm$. Furthermore, the non-trivial relation $\delta^+ \, \delta^- = (\eta_2)^3$ means that $\eta_2$ and $(\eta_2)^2$ are in independent cohomology classes. \Cref{table: psl 3 3} describes the first few cohomology spaces of $\mathfrak{psl}(3|3)$.

\begin{table}[H]
\centering
\begin{tabular}{|c|c|c|}
    \hline\hline
    $\boldsymbol{k}$
        & \textbf{generators of} $\boldsymbol{H^k\bigl( \mathfrak{psl}(3|3) \bigr)}$
        & $\boldsymbol{\beta^k}$
    \\
    \hline\hline
    $1$ & $1$ & $1$
    \\ \hline
    $2$ & $\eta_2$
        & $1$
    \\ \hline
    $3$ & $\psi_3$,
            $\delta^\pm$
        & $3$
    \\ \hline
    $4$ & $(\eta_2)^2$
        & $1$
    \\ \hline
    $5$ & $\psi_5$,
            $\psi_3 \eta_2$, 
            $\psi_3 (\eta_2)^2$
        & $4$
    \\ \hline
    $6$ & $\psi_3 \delta^\pm$, 
            $(\eta_2)^3$
        & $3$
    \\ \hline
    $7$ & $\psi_5 \eta_2$,
            $\psi_3 (\eta_2)^2$,
            $\delta^\pm (\eta_2)^2$
        & $4$
    \\ \hline
    $8$ & $\psi_5 \psi_3$,
            $\psi_5 \delta^\pm$,
            $\psi_3 \delta^\pm \eta_2$,
            $(\delta^\pm)^2 \eta_2$,
            $(\eta_2)^4$
        & $8$
    \\ \hline
    $9$ & $\psi_5 (\eta_2)^2$, 
            $\psi_3 (\eta_2)^3$,
            $\psi_3 (\delta^\pm)^2$,
            $\delta^\pm (\eta_2)^3$
        & $6$
    \\ \hline
    \multicolumn{3}{c}{$\vdots$}
\end{tabular}
\caption{Low order cohomology spaces of $\mathfrak{psl}(3|3)$ in trivial coefficients, their generators and their dimensions.}
\label{table: psl 3 3}
\end{table}

The \nth{3} cohomology of a Lie (super)algebra $\mathfrak{g}$ with coefficients in $V$ is associated with our ability to extend $\mathfrak{g}$ by $V$ to a non-strict skeletal Lie 2-(super)algebra $\mathfrak{g} \ltimes V \rightrightarrows \mathfrak{g}$. More generally, $H^k(\mathfrak{g};V)$ tells us in how many inequivalent ways we can extend $\mathfrak{g}$ by $V$ into a Lie $k$-(super)algebra \cite{Baez-Crans_2010, Huerta_2011}.

Since $\psi_3$ generates \nth{3} cohomology for every $\mathfrak{psl}(n|n)$ (for $n>1$), all of these Lie superalgebras can be extended to Lie 2-superalgebras. From \Cref{table: psl 3 3}, we see that $\mathfrak{psl}(3|3)$ admits \emph{three inequivalent extensions by $\mathbb{K}$} to Lie 2-superalgebras: the other two associated with the cochains $\delta^\pm \colon \mathfrak{psl}(3|3)_1{}^{\wedge3} \longrightarrow \mathbb{K}$.

\subsubsection{Cohomology of \texorpdfstring{$\mathfrak{psl}(4|4)$}{psl(4|4)}}

Finally, let us look at $n=4$. The cohomology ring of $\mathfrak{psl}(4|4)$ with trivial coefficients is generated by the fermionic ghost polynomials $\psi_3$, $\psi_5$ and $\psi_7$, and by the bosonic ghost polynomials $\delta^\pm$ and $\eta_2$ (up to \nth{3} power).

\Cref{table: psl 4 4} gives a basis for the low order cohomology spaces of $\mathfrak{psl}(4|4)$.

\begin{table}[ht]
\centering
\begin{tabular}{|c|c|c|}
    \hline\hline
    $\boldsymbol{k}$
        & \textbf{generators of} $\boldsymbol{H^k\bigl( \mathfrak{psl}(4|4) \bigr)}$
        & $\boldsymbol{\beta^k}$
    \\
    \hline\hline
    $1$ & $1$ & $1$
    \\ \hline
    $2$ & $\eta_2$
        & $1$
    \\ \hline
    $3$ & $\psi_3$
        & $1$
    \\ \hline
    $4$ & $(\eta_2)^2$,
            $\delta^\pm$
        & $3$
    \\ \hline
    $5$ & $\psi_5$,
            $\psi_3 \eta_2$
        & $2$
    \\ \hline
    $6$ & $\delta^\pm \eta_2$, 
            $(\eta_2)^3$
        & $3$
    \\ \hline
    $7$ & $\psi_5 \eta_2$,
            $\psi_3 (\eta_2)^2$,
            $\psi_3 \delta^\pm$
        & $4$
    \\ \hline
    $8$ & $\psi_5 \psi_3$,
            $(\delta^\pm)^2$,
            $\delta^\pm \eta_2$,
            $(\eta_2)^4$
        & $6$
    \\ \hline
    $9$ & $\psi_5 (\eta_2)^2$, 
            $\psi_5 \delta^\pm$,
            $\psi_3 \delta^\pm \eta_2$,
            $\psi_3 (\eta_2)^3$
        & $6$
    \\ \hline
    \multicolumn{3}{c}{$\vdots$}
\end{tabular}
\caption{Low order cohomology spaces of $\mathfrak{psl}(4|4)$ in trivial coefficients, their generators and their dimensions.}
\label{table: psl 4 4}
\end{table}


\subsection{Cohomology of \texorpdfstring{$\mathfrak{sl}(m|n)$}{sl(m|n)} with trivial coefficients} \label{subsec: sl trivial}

It is almost straightforward to see the cohomology of $\mathfrak{sl}(m|n)$ in trivial coefficients too.

For $m=n$, the generators of $H\bigl( \mathfrak{sl}(n|n) \bigr)$ and $H\bigl( \mathfrak{psl}(n|n) \bigr)$ are the same except for $\eta_2$. This means that

\begin{corollary}
The cohomology of $\mathfrak{sl}(n|n)$ in trivial coefficients is given by
\begin{equation} \label{3 cohomology sl m m}
    H\bigl( \mathfrak{sl}(n|n) \bigr)
        \cong H(\mathfrak{sl}_n) \otimes \mathbb{K}[\delta^+ ,\ \delta^-] \big/ \bigl( \delta^+ \, \delta^- = 0 \bigr) .
\end{equation}
\end{corollary}

We can visualize it more clearly in the following way. Let $\hat{\iota}$ be the fermionic ghost associated (through duality and statistics-reversal) with the identity matrix $\mathbb{1}_{2n}$ in $\mathfrak{sl}(n|n)$ --- which is removed when going to $\mathfrak{psl}(n|n)$. It is no surprise that $\eta_2$ is the image of $\hat{\iota}$ under the BRST differential on $\mathfrak{gl}(n|n)$, viz.
\begin{equation*}
    Q^{\mathfrak{gl}} (\hat{\iota})
        = \eta_2 .
\end{equation*}
Since $\eta_2$ is exact in $\mathfrak{gl}(n|n)$, it is obviously closed.

Going to $\mathfrak{sl}(n|n)$ and later to $\mathfrak{psl}(n|n)$, $\eta_2$ remains closed because the BRST differential is carried over through \eqref{2 Q psl vs Q gl}. By removing $\mathbb{1}_{2n}$ (and its associated ghost $\hat{\iota}$) to form $\mathfrak{psl}(n|n)$, the polynomial $\eta_2$ is no longer exact, thus generating \nth{2} cohomology for $\mathfrak{psl}(n|n)$.

\bigskip

The most interesting example is $n=1$, as $\mathfrak{sl}_1$ has nontrivial cohomology only in degree $0$. Then
\begin{equation*}
    H\bigl( \mathfrak{sl}(1|1) \bigr)
        \cong \mathbb{K}[\delta^+,\delta^-] \big/\bigl( 
        \delta^+ \, \delta^- = 0 \bigr)
        \cong \mathbb{K}[\gamma^+] \oplus \mathbb{K}[\gamma^-] ,
\end{equation*}
i.e. we have two independent ``towers'' of powers of each determinant.

\bigskip

When $m \neq n$, the invariant polynomials are precisely the same as those for $\mathfrak{gl}(m|n)$. The generator of the \nth{1} cohomology of $\mathfrak{gl}(m|n)$, namely $\psi_1^{} = \varphi'_1 - \varphi''_1$, is associated with the supertrace in $\mathfrak{gl}(m|n)$. By going to $\mathfrak{sl}(m|n)$, we set $\psi_1$ to zero; thus we have $\varphi'_1 + \varphi''_1$ transgressing to $2\eta_2$. No \nth{1} cohomology is left for $\mathfrak{sl}(m|n)$, recovering the result of \textcite{Fuks_1986}:

\begin{corollary}
When $m\neq n$, the cohomology of $\mathfrak{sl}(m|n)$ in trivial coefficients is given by
\begin{equation}
    H\bigl( \mathfrak{sl}(m|n) \bigr)
        \cong H \bigl( \mathfrak{sl}_{\max\{m,n\}} \bigr) .
\end{equation}
\end{corollary}

We would like to stress that Theorem 2.6.1 in \cite{Fuks_1986} is correct for $\mathfrak{gl}(m|n)$ for any $m,n\in\mathbb{N}$, and correct for $\mathfrak{sl}(m|n)$ \emph{as long as} $m \neq n$. When $m=n$, which is physically more relevant as discussed in \Cref{sec: intro}, we also have bosonic ghost polynomials generating of cohomology.

\section{Cohomology in the adjoint representation} \label{sec: adjoint coefficients}

In this Section, we will compute the cohomology of $\mathfrak{psl}(n|n)$ --- henceforth simply denoted $\mathfrak{g}$ --- in the adjoint representation, viz.
\begin{equation*}
    H^\bullet_{\mathrm{ad}} ( \mathfrak{g} )
        \equiv H^\bullet ( \mathfrak{g} ; \mathfrak{g} ) .
\end{equation*}
As before, to keep track of each bosonic subsector of $\mathfrak{g}$ (and their index structure), we will postpone the imposition of $m=n$ till later. This means that, in the next Subsection, we will effectively work with\footnote{
    This does not apply to $\mathfrak{sl}(m|n)$ directly, since there is a missing $\mathfrak{u}_1$ part in $\mathfrak{g}_0$; see \eqref{2 isomorphism sl m n 0}.
}
\begin{align*}
    \mathfrak{g}_0
        &\coloneqq \mathfrak{sl}_m \oplus \mathfrak{sl}_n ,
    \\
    \mathfrak{g}_1
        &\coloneqq U^+ \oplus U^- \equiv U .
\end{align*}

We will do as follows.
\begin{itemize}
    \item In \Cref{subsec: covariants}, we construct explicitly the covariant maps that appear in the \nth{1} page of the HS spectral sequence for $(\mathfrak{g},\mathfrak{g}_0;\mathrm{ad})$. This computation is split into three parts.

    \item In \Cref{subsec: 1st page HS adjoint}, we give the full description of the module of covariant maps over invariant polynomials, analyzing their non-trivial relations.

    \item In order to compute the next pages of the HS spectral sequence, we need to understand how the HS differentials act; this is done in detail in \Cref{subsec: BRST on covariants}.
    
    \item In \Cref{subsec: 2nd page HS adjoint}, we compute the \nth{2} and higher pages of the HS spectral sequence, arriving at the final result for the cohomology of $\mathfrak{psl}(n|n)$ in adjoint coefficients.
\end{itemize}

\subsection{\texorpdfstring{$\mathfrak{g}_0$}{g0}-invariant \texorpdfstring{$\mathfrak{g}$}{g}-valued polynomials in ghost coordinates of \texorpdfstring{$\mathfrak{g}_1$}{g1}} \label{subsec: covariants}

In order to use the Hochschild--Serre spectral sequence \eqref{2 def spectral sequence pages}, we need to understand the ``$\mathfrak{g}_0$-invariant $\mathfrak{g}$-valued polynomials in ghost coordiantes of $\mathfrak{g}_1$'' --- or, using the terminology of \Cref{sec: invariant theory}, the $\mathfrak{g}_0$-covariants of $\mathbb{K}[U]$ of type $\mathfrak{g}$.

If we let
\begin{equation*}
    \mathcal{I}
        \coloneqq \operatorname{Inv}_{\mathfrak{g}_0}
            \operatorname{Hom}\bigl( \mathrm{S}^\bullet(U) , \mathfrak{g} \bigr)
        \cong \operatorname{Inv}_{\mathfrak{g}_0} \bigl( \mathbb{K}[U] \otimes \mathfrak{g} \bigr) ,
\end{equation*}
the \nth{1} page of the HS spectral sequence for $(\mathfrak{g},\mathfrak{g}_0;\mathrm{ad})$ is
\begin{equation*}
    E_1^{p,q}
        \cong H^q (\mathfrak{g}_0) \otimes \mathcal{I}^p ,
\end{equation*}
where $\bigoplus_p \mathcal{I}^p$ is the natural $\mathbb{N}_0$-grading of $\mathcal{I}$ by polynomial degree.

The Lie superbracket is a map of degree $0$, so the action of $\mathfrak{g}_0$ on $\mathfrak{g} = \mathfrak{g}_0 \oplus \mathfrak{g}_1$ does not mix components, and we can make the further decomposition
\begin{align*}
    \mathcal{I}
        \cong \mathcal{I}_{0} \oplus \mathcal{I}_1 ,
    \qquad\text{where}\quad
    \mathcal{I}_i
        \coloneqq \operatorname{Inv}_{\mathfrak{g}_0} \bigl( \mathbb{K}[U] \otimes \mathfrak{g}_i \bigr)
        \quad \text{and} \quad i\in\mathbb{Z}_2 .
\end{align*}

\bigskip

As one might expect from the discussion about trivial coefficients in \Cref{subsec: invariant polynomials}, we will have trace-like and determinant-like covariants. Since the latter are special to $\mathfrak{sl}(m|n)$ only when $m=n$, they are presented last. Indeed, it all boils down to the following.

\begin{theorem}[tensor FFT for $\mathrm{SL}_n$ \cite{Erickson_2025, Schrijver_2008}]
Let $V = \mathbb{K}^n$ and $\mathcal{T}^{p,q}(V) = V^{\otimes p} \otimes (V^\ast)^{\otimes q}$ be the space of mixed tensors of type $(p,q)$ over $V$. A basis for $\bigl( \mathcal{T}^{p,q}(V) \bigr)^{\operatorname{SL}_n}$ is composed of tensors whose coefficients are products, $(\mathcal{S}_p \times \mathcal{S}_q)$-permutations and contractions of
\begin{align*}
    \delta\indices{^i_j}  ,
    \qquad \epsilon\indices{^{i_1}^\ldots^{i_n}} ,
    \qquad \epsilon\indices{_{j_1}_\ldots_{j_n}} ,
\end{align*}
with all indices ranging in $\{1,\ldots,n\}$.
\end{theorem}

Mind that there is a non-trivial relation between the Kronecker deltas and the Levi-Civita symbols, namely
\begin{equation*}
    \epsilon\indices{^{i_1}^\ldots^{i_n}} \, \epsilon\indices{_{j_1}_\ldots_{j_n}}
        = \begin{vmatrix}
            \delta\indices{^{i_1}_{j_1}}
                & \cdots
                & \delta\indices{^{i_1}_{j_n}}
            \\
            \vdots
                & \ddots
                & \vdots
            \\
            \delta\indices{^{i_n}_{j_1}}
                & \cdots
                & \delta\indices{^{i_n}_{j_n}}
        \end{vmatrix}
\end{equation*}
--- which is the analogue of \eqref{3 SFT e} for tensors.

In the following Subsubsections, we will give an intuitive approach on how to find these invariant tensors. The tensors whose coefficients are entirely in terms of Kronecker deltas are called \emph{trace-like}; the remaining are called \emph{determinant-like}.

\subsubsection{Description of \texorpdfstring{$\mathcal{I}_1^p$}{I1p}}

Let us first take a look at maps of the type $\mathrm{S}^p (U) \longrightarrow \mathfrak{g}_1 \equiv U$.

Given the index structure of $U = U^+ \oplus U^-$ in \eqref{2 index structure of U}, it is impossible to send two elements of $U$ and to another one in $U$, i.e.
\begin{equation*}
    \mathcal{I}_1^{2k} = 0 .
\end{equation*}

For $p=1$, there are the two $\mathfrak{g}_0$-covariant maps:
\begin{equation} \label{4 ad-invariant i=1 p=1}
    \gamma^\pm \colon (u^+,u^-) \longmapsto u^\pm .
\end{equation}
Up to linear combinations, these are the only elements in $\mathcal{I}_1^1 \cong \operatorname{Hom}(U,U)^{\mathfrak{g}_0}$. It is easy to convince oneself once we write $U = U^+ \oplus U^-$ and apply Schur's Lemma on each component of a map $U \longrightarrow U$ \cite{Fulton-Harris_2004}.

For $p=3$, the maps $\eta_2 \gamma^\pm$ are $\mathfrak{g}_0$-covariants because $\eta_2 \equiv \operatorname{tr}(\gamma^+\gamma^-)$ is an invariant polynomial, and $\gamma^\pm$ are covariant maps. Besides, the maps
\begin{equation*}
    \gamma^\pm \gamma^\mp \gamma^\pm
        \colon (u^+,u^-) \longmapsto u^\pm u^\mp u^\pm
\end{equation*}
are also covariants. This feature will happen again at each step, hence
\begin{subequations} \label{4 I1 contains R-combinarions of odd powers}
\begin{align}
    \mathcal{I}_1
        &\supseteq \mathcal{R} \bigl\{
            \gamma^\pm ,\
            \gamma^\pm \gamma^\mp \gamma^\pm ,\
            \ldots \bigr\} ,
    \\
    \mathcal{I}_1^{2s+1}
        &\supseteq \bigoplus_{t=0}^s \mathcal{R}^{2(s-t)} \bigl\{ (\gamma^\pm \gamma^\mp)^t \, \gamma^\pm \bigr\} .
\end{align}
\end{subequations}

Insofar as trace-like polynomials are considered (for which we will use a tilde), the inclusion $\supseteq$ in \eqref{4 I1 contains R-combinarions of odd powers} is an equality. Indeed, consider we want a map of order $(2s+1)$ that cannot be split into a product of lower-degree map with invariant polynomials\footnote{They shall be called \emph{pure covariants}, in contrast to the \emph{composite covariants}.}; then we must have all $2s+1$ matrix entries contracted together, which must be done with Kronecker deltas.

Explicitly, for small values of $p$,
\begin{alignat*}{5}
    \widetilde{\mathcal{I}}_1^1
        &= \mathbb{K}\bigl\{ \gamma^\pm \bigr\}
        &&
        &&
        &&\cong \mathbb{K}^2 ,
    \\
    \widetilde{\mathcal{I}}_1^3
        &= \mathbb{K}\bigl\{
                M \, \gamma^- ,\ 
                N \, \gamma^+ \bigr\}
            \;&&\oplus\; \eta_2 \cdot \widetilde{\mathcal{I}}_1^1
            &&
        &&\cong \mathbb{K}^4 ,
    \\
    \widetilde{\mathcal{I}}_1^5
        &= \mathbb{K}\bigl\{
                M^2 \, \gamma^- ,\
                N^2 \, \gamma^+ \bigr\}
            \;&&\oplus\; \eta_2 \cdot \widetilde{\mathcal{I}}_1^3 
            \;&&\oplus\; \eta_4 \cdot \widetilde{\mathcal{I}}_1^1
        &&\cong \mathbb{K}^8 ,
    \\
    \widetilde{\mathcal{I}}_1^7
        &= \mathbb{K}\bigl\{
                M^3 \, \gamma^- ,\
                N^3 \, \gamma^+ \bigr\}
            \;&&\oplus\; \eta_2 \cdot \widetilde{\mathcal{I}}_1^5 
            \;&&\oplus\; \eta_4 \cdot \widetilde{\mathcal{I}}_1^3
            \;\oplus\; \eta_6 \cdot \widetilde{\mathcal{I}}_1^1
        \quad&&\cong \mathbb{K}^{16} ,
    \\
    \text{etc.},
\end{alignat*}
where
\begin{align} \label{4 def M and N}
    M \equiv \gamma^- \gamma^+ \in \mathfrak{gl}_m
    \qquad\text{and}\qquad
    N \equiv \gamma^+ \gamma^- \in \mathfrak{gl}_n . 
\end{align}

While $s\leqslant\min\{m,n\}$, we have new invariants $\eta_{2s}$ that contribute to new composite covariants; after that, they are no longer algebraically independent, but we can still construct higher order covariants in that way, hence
\begin{equation*}
    \dim \widetilde{\mathcal{I}}_1^{2s+1}
        = 2^{s+1} .
\end{equation*}

\subsubsection{Description of \texorpdfstring{$\mathcal{I}_0^p$}{I0p}}

For the maps of type $\mathrm{S}^p(U) \longrightarrow \mathfrak{g}_0$, we can further split $\mathfrak{g}_0 \cong \mathfrak{sl}_m \oplus \mathfrak{sl}_n$ and focus only in one sector, say
    $\mathrm{S}^p(U) \longrightarrow \mathfrak{sl}_m$ .
Everything works analogously for $\mathfrak{sl}_n$.

Similarly to the previous Subsection, there is no way of arranging an odd number of $\gamma^+$'s and $\gamma^-$'s and land on $\mathfrak{sl}_m$, thus
\begin{equation}
    \mathcal{I}_0^{2s+1} = 0 .
\end{equation}
The only way to do so is by taking $M \equiv \gamma^- \gamma^+ \in \mathfrak{gl}_m$ (and higher powers) and projecting it onto $\mathfrak{sl}_m$ by removing its trace. Namely, if
\begin{alignat*}{3}
    \pi_0 \colon \mathfrak{gl}(V_0) &\longtwoheadrightarrow \mathfrak{sl}(V_0) 
    \\
    M &\longmapsto M - \tfrac{1}{m} \operatorname{tr}(M) \, \mathbb{1}_m
\end{alignat*}
is said projection, the maps
\begin{alignat*}{3}
    \pi_0(M^s) 
        = M^s - \tfrac{1}{m} \operatorname{tr}(M^s) \, \mathbb{1}_m
\end{alignat*}
are covariants for every $s\in\mathbb{N}$.

If $\pi_1 \colon \mathfrak{gl}(V_1) \longtwoheadrightarrow \mathfrak{sl}(V_1)$ is the other projection, then
\begin{equation}
    \mathcal{I}_0
        \supseteq \mathcal{R} \bigl\{ \pi_0(M^s) , \pi_1(N^s) \, \big| \, s\in\mathbb{N} \bigr\} .
\end{equation}
Again, as far as we consider trace-like covariants, the above inclusion is an equality. For small $p$, we have
\begin{alignat*}{5}
    \widetilde{\mathcal{I}}_0^2
        &= \mathbb{K}\bigl\{
                \pi_0(M) ,\
                \pi_1(N) \bigr\}
            &&
        &&\cong \mathbb{K}^2 ,
    \\
    \widetilde{\mathcal{I}}_0^4
        &= \mathbb{K}\bigl\{
                \pi_0(M^2) ,\ 
                \pi_1(N^2) \bigr\}
            \;&&\oplus\;
            \eta_2 \cdot \widetilde{\mathcal{I}}_0^2
        &&\cong \mathbb{K}^4 ,
    \\
    \widetilde{\mathcal{I}}_0^6
        &= \mathbb{K}\bigl\{
                \pi_0(M^3) ,\ 
                \pi_1(N^3) \bigr\}
            \;&&\oplus\;
            \eta_2 \cdot \widetilde{\mathcal{I}}_0^4
            \;\oplus\;
            \eta_4 \cdot \widetilde{\mathcal{I}}_0^2
        &&\cong \mathbb{K}^8 ,
    \\
    \widetilde{\mathcal{I}}_0^8
        &= \mathbb{K}\bigl\{
                \pi_0(M^4) ,\ 
                \pi_1(N^4) \bigr\}
            \;&&\oplus\;
            \eta_2 \cdot \widetilde{\mathcal{I}}_0^6
            \;\oplus\;
            \eta_4 \cdot \widetilde{\mathcal{I}}_0^4
            \;\oplus\;
            \eta_6 \cdot \widetilde{\mathcal{I}}_0^2
        \quad&&\cong \mathbb{K}^{16} ,
    \\
    \text{etc.}
\end{alignat*}
and
\begin{align*}
    \dim \widetilde{\mathcal{I}}_0^{2s} = 2^s .
\end{align*}

\subsubsection{Determinant-like covariants}

Recall that the determinant-like polynomials $\delta^\pm$ are elements of $\mathcal{R} \equiv \operatorname{Inv}_{\mathfrak{g}_0} \mathbb{K}[U]$ only when $m=n$. In that case, we will have two special covariants of degree $n-1$: the \emph{adjugate matrices} $\mu^\pm$ of $\gamma^\pm$, defined by
\begin{align} \label{4 def adjugate}
    \mu^\pm \gamma^\pm
    = \gamma^\pm \mu^\pm    
        = \det(\gamma^\pm) \, \mathbb{1}_n .
\end{align}
Since $\gamma^\pm \in U^\pm$, we must have $\mu^\pm \in U^\mp$. 

The RHS of \eqref{4 def adjugate} is $\mathfrak{g}_0$-invariant, because $\det(\gamma^\pm)$ is a $\mathfrak{g}_0$-invariant polynomial. Taking $g\in\mathrm{SL}(V_0)$ and $h\in\mathrm{SL}(V_1)$, we know that $(g,h) \vartriangleright \gamma^+ = h \gamma^+ g^{-1}$; it is straightforward
to check that
\begin{equation*}
    (g,h) \vartriangleright \mu^+
        = g \, \mu^+ \, h^{-1} .
\end{equation*}
But that is precisely the action of $(g,h)$ on an element of $U^-$, thus the map
\begin{align*}
    \mu^+ \colon \mathrm{S}^{n-1}(U^+) &\longrightarrow U^-
\end{align*}
is a $\mathfrak{g}_0$-covariant. Similarly for $\mu^-$.

\subsection[The \texorpdfstring{\nth{1}}{1st} page of HS spectral sequence]{The \texorpdfstring{\nth{1}}{1st} page of the Hochschild--Serre spectral sequence} \label{subsec: 1st page HS adjoint}

From the discussion in the previous Subsection, we can characterize the space $\mathcal{I} = \bigoplus_{p>0} \operatorname{Inv}_{\mathfrak{g}_0} \operatorname{Hom}\bigl( \mathrm{S}^p (U) , \mathfrak{g} \bigr)$ as the space of $\mathcal{R}$-linear combinations of the covariant maps
\begin{subequations} \label{4 covariant maps}
\begin{alignat}{4}
    \pi_0 (\gamma^- \gamma^+)^s
        &\colon&
        \mathrm{S}^{2s}(U) &\longrightarrow \mathfrak{sl}(V_0) ,
    \\
    \pi_1 (\gamma^+ \gamma^-)^s
        &\colon&
        \mathrm{S}^{2s}(U) &\longrightarrow \mathfrak{sl}(V_1) ,
    \\
    \Gamma_\pm^s \equiv (\gamma^\pm \gamma^\mp)^s \, \gamma^\pm
        &\colon&
        \mathrm{S}^{2s+1}(U) &\longrightarrow U ,
    \\
    \mu^\pm
        &\colon&
        \mathrm{S}^{n-1}(U^\pm) &\longrightarrow U^\mp ,
\end{alignat}
\end{subequations}
for $s\in\mathbb{N}$.

The $q=0$ row of the HS spectral sequence for $\mathfrak{psl}(n|n)$ in adjoint coefficients is given by $\mathcal{I}^p$. The $p=0$ column is, as before, ``given'' by the cohomology ring $H^q(\mathfrak{g}_0)$ --- but there are no \nth{0} degree covariants. From \eqref{2 1st page decomposition}, we have the full description of $E_1$ as the bi-graded tensor product of $H^\bullet(\mathfrak{g}_0)$ and $\mathcal{I}^\bullet$. This tensor structure is preserved by the differentials, so that we can focus only on the generators aforementioned.

\Cref{fig2} shows a schematization of the \nth{1} page with all the generators, \emph{up to coefficients in $\mathcal{R}$}. Since the elements of $\mathcal{R}$ are $\mathfrak{g}_0$-invariant polynomials, they are in the kernel of the differentials and do not affect transgressions.

\begin{figure}[!ht]
\centering
\begin{tikzpicture}
\begin{axis}[
    axis lines = middle,
    axis line style = {->, >=Stealth},
    xlabel = {$p$},
    ylabel = {$q$},
    x label style = {anchor = west},
    y label style = {anchor = south},
    xmin = -0.5,    xmax = 9,
    ymin = -0.5,    ymax = 10,
    xtick = {1, 2, 3, 4, 7},
    xticklabels = {$1$, $2$, $3$, $4$, $n-1$},
    ytick = {1, 3, 7, 9},
    yticklabels = {$1$, $3$, $2n-3$, $2n-1$},
    grid = none,
    axis x line = bottom,
    axis y line = left,
    x axis line style = {draw=none},
    y axis line style = {draw=none},
    label style = {font=\small},
]


\draw (axis cs:-1,0) -- (axis cs:4.9,0);
\draw[->, >=Stealth] (axis cs:5.1,0) -- (axis cs:9,0);
\draw[] (axis cs:4.9,-0.2) -- (axis cs:4.9,+0.2);
\draw[] (axis cs:5.1,-0.2) -- (axis cs:5.1,+0.2);

\draw (axis cs:0,-1) -- (axis cs:0,4.9);
\draw[->, >=Stealth] (axis cs:0,5.1) -- (axis cs:0,10);
\draw[] (axis cs:-0.2,4.9) -- (axis cs:+0.2,4.9);
\draw[] (axis cs:+0.2,5.1) -- (axis cs:-0.2,5.1);

\filldraw[black, fill=white] (0,1.1) circle (1.5pt);
    \draw (0,1.4) node[anchor=west] {\scriptsize$\varphi'_1$};
\filldraw[black, fill=white] (0,0.9) circle (1.5pt);
    \draw (0,0.6) node[anchor=west] {\scriptsize$\varphi''_1$};

\filldraw[black, fill=white] (0,3.1) circle (1.5pt);
    \draw (0,3.4) node[anchor=west] {\scriptsize$\varphi'_3$};
\filldraw[black, fill=white] (0,2.9) circle (1.5pt);
    \draw (0,2.6) node[anchor=west] {\scriptsize$\varphi''_3$};

\filldraw[black, fill=white] (0,7.1) circle (1.5pt);
    \draw (0,7.4) node[anchor=west] {\scriptsize$\varphi'_{2n-3}$};
\filldraw[black, fill=white] (0,6.9) circle (1.5pt);
    \draw (0,6.6) node[anchor=west] {\scriptsize$\varphi''_{2n-3}$};

\filldraw[black, fill=white] (0,9.1) circle (1.5pt);
    \draw (0,9.4) node[anchor=west] {\scriptsize$\varphi'_{2n-1}$};
\filldraw[black, fill=white] (0,8.9) circle (1.5pt);
    \draw (0,8.6) node[anchor=west] {\scriptsize$\varphi''_{2n-1}$};

\filldraw[cyan] (0.9,0) circle (1.5pt);
\filldraw[cyan] (1.1,0) circle (1.5pt);
    \draw (1,0.5) node[cyan] {\footnotesize$\Gamma_\pm^1$};

\filldraw[black] (1.9,0) circle (1.5pt);
\filldraw[black] (2.1,0) circle (1.5pt);
    \draw (2,0.5) node {\footnotesize$\pi_{0/1}^1$};

\filldraw[cyan] (2.9,0) circle (1.5pt);
\filldraw[cyan] (3.1,0) circle (1.5pt);
    \draw (3,0.5) node[cyan] {\footnotesize$\Gamma_2^\pm$};

\filldraw[black] (3.9,0) circle (1.5pt);
\filldraw[black] (4.1,0) circle (1.5pt);
    \draw (4,0.5) node {\footnotesize$\pi_{0/1}^2$};

\filldraw[black] (5.9,0) circle (1.5pt);
\filldraw[black] (6.1,0) circle (1.5pt);

\filldraw[black] (6.9,-0.125) circle (1.5pt);
\filldraw[black] (7.1,-0.125) circle (1.5pt);
\filldraw[red] (6.9,0.125) circle (1.5pt);
\filldraw[red] (7.1,0.125) circle (1.5pt);
    \draw (7,0.5) node[red] {\footnotesize$\mu_\pm$};

\filldraw[black] (7.9,0) circle (1.5pt);
\filldraw[black] (8.1,0) circle (1.5pt);


\end{axis}
\end{tikzpicture}
\caption{\label{fig2}Schematization of the \nth{1} page of the Hochschild--Serre spectral sequence for $\mathfrak{psl}(n|n)$ in the adjoint representation; full dots represent the generators of $\mathcal{I}$ up to coefficients in $\mathcal{R}$ --- where $\pi_0^k$ is a shorthand notation for $\pi_0(M^k)$ etc. The fermionic ghost generators $\varphi'_{2k-1},\varphi''_{2k-1} \in E_1^{0,2k-1}$ are the same as before; however, as there is no \nth{0} order covariant map, they themselves are not present in the \nth{1} page, but their products with the covariants are. The bosonic ghost generators (the covariant maps) appear in pairs for each $E_1^{p,0}$, except for $p = n-1$, where the adjugate matrices also appear.}
\end{figure}
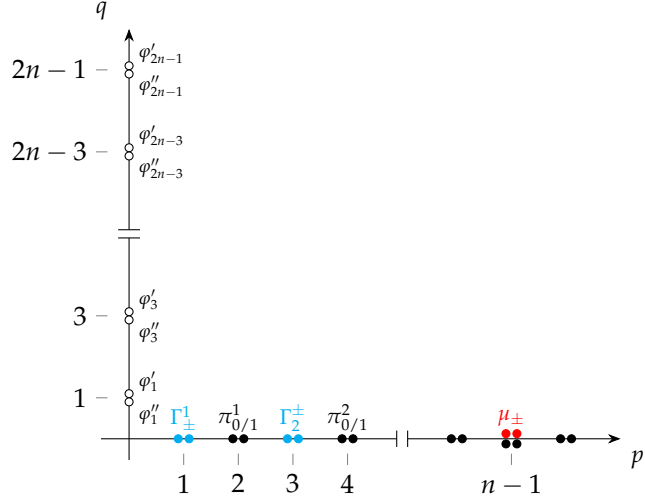

The space $E_1^{2n-1,0}$ is composed of the odd-degree trace-like covariants $\Gamma^\pm_{n-1}$, but also lower degree covariants multiplied by invariant polynomials in $\mathcal{R}$. However, $\mathcal{I}$ is not a free module over $\mathcal{R}$.

\begin{proposition} \label{prop delta mu}
The composite covariants $\delta^\pm \mu^\mp$ are $\widetilde{\mathcal{R}}$-linear combinations of the trace-like pure covariants, with $\widetilde{\mathcal{R}} \coloneqq \mathbb{K}\bigl[ \eta_{2} ,\ \ldots ,\ \eta_{2n} \bigr]$. 
\end{proposition}
\begin{proof}
Consider the $n \times n$ matrix $M \equiv \gamma^- \gamma^+ \in \mathfrak{gl}(V_0)$ and define the function
\begin{equation*}
    f \equiv \Gamma\indices*{^-_{n-1}}
        + c\indices{_{n-1}} \Gamma\indices*{^-_{n-2}}
        + \cdots
        + c\indices{_1} \Gamma\indices*{^-_0}
        + (-1)^n \, \delta^- \mu^+ 
        \colon \mathrm{S}^{n-1}(U) \longrightarrow U^- ,
\end{equation*}
where the functions $c_s$ are the coefficients of the characteristic polynomial of $M \equiv \gamma^- \gamma^+$. In the light of Newton's identities, they can be written in terms of the trace-like invariant polynomials $\eta_{2s} = \operatorname{tr}(M^s)$, for $s\in\{1,\ldots,n\}$ --- i.e. elements of $\widetilde{\mathcal{R}}$.
\par\noindent
Since $M^s = \Gamma^-_{s-1} \gamma^+$ and $\det(M) = \delta^+ \delta^-$, the polynomial
\begin{equation*}
    f \gamma^+
        = M^n + c_{n-1} M^{n-1} + \cdots + c_1 M + (-1)^n \det(M) \, \mathbb{1}_n
\end{equation*}
vanishes identically due to the Cayley--Hamilton theorem.
\par\noindent
Thinking of $\gamma^+$ as a matrix, we have $f = f \gamma^+ (\gamma^+)^{-1} = 0$ whenever $\gamma^+$ is invertible. But every square matrix in $\mathbb{K}$ is the limit of a sequence of invertible matrices, so $f=0$ everywhere.
\end{proof}

\subsection{The BRST differential on the covariants} \label{subsec: BRST on covariants}

As mentioned in \Cref{subsec: spectralseq}, the differentials on the Hochschild--Serre spectral sequence are lifts of the BRST differential to the associated filtration and its cohomologies, schematically given by \eqref{2 dr schematic}. For that reason, it is useful to understand how $Q_{\operatorname{ad}}$ acts on the covariant maps \eqref{4 covariant maps}.

\begin{proposition} \label{prop: BRST on Gamma pm}
The BRST operator of $\mathfrak{psl}(n|n)$ in adjoint coefficients acts on the odd-degree covariants $\Gamma^\pm_s \colon \mathrm{S}^{2s+1}(U) \longrightarrow U^\pm$ as
\begin{align} \label{4 Qad Gamma pm s}
    Q^\mathfrak{psl}_{\operatorname{ad}}\bigl( \Gamma^\pm_s \bigr)
        = \pi_0(M^{s+1}) + \pi_1(N^{s+1}) .
\end{align}
\end{proposition}
\begin{proof}
Let us start with $s=0$, i.e. with the identity maps $\gamma^\pm \colon U^\pm \longrightarrow U^\pm$. Using the index convention introduced in \eqref{2 index convention} 
and by virtue of \eqref{2 Theta on coordinates}, we have
\begin{subequations} \label{4 Theta gamma pm}
\begin{align} 
    \Theta(\gamma^+)
        &= \bigl( b\indices{^I_L} \, \gamma\indices{^L_j} - \gamma\indices{^I_\ell} \, a\indices{^\ell_j} \bigr) \otimes \tau\indices{_I^j}
        + M\indices{^i_j} \otimes \tau\indices{_i^j}
        + N\indices{^I_J} \otimes \tau\indices{_I^J} ,
    \\
    \Theta(\gamma^-)
        &= \bigl( a^i{}_{\!\ell} \, \gamma^\ell{}_{\!J} - \gamma^i{}_{\!L} \, b^L{}_{\!J} \bigr) \otimes \tau_i{}^J
        + M^i{}_j \otimes \tau_i{}^j
        + N^I{}_J \otimes \tau_I{}^J .
\end{align}
\end{subequations}
The terms in parentheses are precisely the negative of the ones that come from $Q(\gamma^\pm)$ in \eqref{2 BRST gl with indices}, hence
\begin{align*}
    Q_{\operatorname{ad}}(\gamma^+)
        = Q_{\operatorname{ad}}(\gamma^-)
        = M + N .
\end{align*}
Similarly, for $s\in\mathbb{N}$, we have
\begin{align*}
    \Theta(\Gamma^+_s)
        &= b \, \Gamma^+_s - \Gamma^+_s a
            + M^{s+1} + N^{s+1} ,
    \\
    \Theta(\Gamma^-_s)
        &= a \, \Gamma^-_s - \Gamma^-_s b
            + M^{s+1} + N^{s+1} ;
\end{align*}
in turn, we can write $\Gamma^+_s = N^s \gamma^+$, so
\begin{align*}
    Q(\Gamma^+_s)
        &= - b \, \Gamma^+_s + \Gamma^+_s \, a ,
    \\
    Q(\Gamma^-_s)
        &= - a \, \Gamma^-_s + \Gamma^-_s \, b ,
\end{align*}
thus
\begin{align*}
    Q_{\operatorname{ad}}(\Gamma^\pm_s)
        = M^{s+1} + N^{s+1} .
\end{align*}
Up until now, both $Q$ and $\Theta$ were part of the BRST operator of $\mathfrak{gl}(n|n)$; the projection \eqref{2 Q psl vs Q gl} to $\mathfrak{psl}(n|n)$ yields \eqref{4 Qad Gamma pm s}, as we wanted.
\end{proof}

According to \eqref{4 Qad Gamma pm s}, the identity map $\gamma \equiv \gamma^+ + \gamma^-$ on $U$ is not even closed, but $\gamma^+ - \gamma^-$ is. Thus, we already expect that $\gamma^+ - \gamma^-$ will generate $H^1(\mathfrak{g};\mathfrak{g})$.

\begin{proposition} \label{prop: BRST on pi M and pi N}
The BRST operator of $\mathfrak{psl}(n|n)$ in adjoint coefficients acts on the even-degree covariants $\pi_0(M^s) \colon \mathrm{S}^{2s}(U) \longrightarrow \mathfrak{sl}(V_0)$ and $\pi_1(N^s) \colon \mathrm{S}^{2s}(U) \longrightarrow \mathfrak{sl}(V_1)$ as
\begin{subequations} \label{4 Qad Ms Ns}
\begin{align}
    Q^\mathfrak{psl}_{\operatorname{ad}}\bigl( \pi_0(M^s) \bigr)
        &= \Gamma_s^+ - \Gamma_s^- ,
    \\
    Q^\mathfrak{psl}_{\operatorname{ad}}\bigl( \pi_1(N^s) \bigr)
        &= \Gamma_s^- - \Gamma_s^+ .
\end{align}
\end{subequations}
\end{proposition}
\begin{proof}
We start by evaluating how the coordinates of $M$ and $N$ behave under the action of $Q$:
\begin{subequations} \label{4 Q M N}
\begin{alignat}{4}
    Q M\indices{^i_j}
        &= M\indices{^i_\ell} \, a\indices{^\ell_j}
            - a\indices{^i_\ell} \, M\indices{^\ell_j}
        \quad&&\Longrightarrow\quad
        QM &= [M,a] ,
    \\
    Q N\indices{^I_J}
        &= N\indices{^I_K} \, b\indices{^K_J}
            - b\indices{^I_L} \, N\indices{^L_J}
        \quad&&\Longrightarrow\quad
        QN &= [N,b] .
\end{alignat}
\end{subequations}
Using \eqref{2 Theta on coordinates}, we have
\begin{align*}
    \Theta(M^s)
        &= [a,M^s] + \gamma^+ M^s - M^s \gamma^- ,
    \\
    \Theta(N^s)
        &= [b,N^s] + \gamma^- N^s - N^s \gamma^+ .
\end{align*}
The terms $[M^s,a]$ and $[N^s,b]$ are precisely the action of the BRST operator in trivial coefficients on $M$ and $N$. Indeed, we know that this is true for $s=1$, as we showed in \eqref{4 Q M N}; by induction, we see that
\begin{align*}
    Q(M^s)
        &= Q(M^{s-1}) \, M + M^{s-1} \, Q(M)
        \\
        &= [M^{s-1},a] \, M + M^{s-1} \, [M,a]
        = [M^s,a] ,
\end{align*}
as the bracket above is, in fact, a supercommutator --- see \eqref{2 supercommutator}.
\par\noindent
Since $\gamma^+ M^s = N^s \gamma^+ = \Gamma^+_s$ and similarly to $\Gamma_s^-$, it follows that
\begin{align*}
    Q_{\operatorname{ad}} (M^s)
        &= \Gamma_s^+ - \Gamma_s^- ,
    \\
    Q_{\operatorname{ad}} (N^s)
        &= \Gamma_s^- - \Gamma_s^+ .
\end{align*}
As the RHS is $\mathfrak{g}_1$-valued, it is unaffected by the projection \eqref{2 Q psl vs Q gl}. Therefore \eqref{4 Qad Ms Ns} follows, as we wanted.
\end{proof}

\begin{proposition} \label{prop: adjugates are Q-closed for psl}
The adjugate covariants $\mu^\pm \colon \mathrm{S}^{n-1}(U^\pm) \longrightarrow U^\mp$ are closed under the BRST operator of $\mathfrak{psl}(n|n)$ in adjoint coefficients, i.e.
\begin{equation}
    Q^\mathfrak{psl}_{\operatorname{ad}}(\mu^\pm)
        = 0 .
\end{equation}
\end{proposition}
\begin{proof}
It suffices to consider only $\mu^+$, which reads in coordinates as
\begin{equation*}
    \mu^+
        = \epsilon\indices{^{i_1}^\ldots^{i_n}} \epsilon\indices{_{J_1}_\ldots_{J_n}} \,
        \gamma\indices{^{J_1}_{i_1}} \, \cdots \, \gamma\indices{^{J_{n-1}}_{i_{n-1}}} \otimes \tau\indices{_{i_n}^{J_n}} .
\end{equation*}
For $n=2$, we have
\begin{align*}
    \mu^+
        &= \epsilon\indices{^i^j} \epsilon\indices{_I_J} \gamma\indices{^I_i} \otimes \tau\indices{_j^J} ,
    \\
    Q_{\operatorname{ad}} (\mu^+)
        &= \Bigl(
            \epsilon\indices{^i^j} \, b\indices{^K_{[I}} \epsilon\indices{_{J]}_K}
            - \epsilon\indices{_I_J} \, \epsilon\indices{^k^{[i}} a\indices{^{j]}_k}
        \Bigr) \, \gamma^J{}_j \otimes \tau\indices{_i^I}
        + \det(\gamma^+) \, \mathbb{1}_{2n} ,
\end{align*}
where $S\indices{^{[i}} T\indices{^{j]}} \coloneqq S^i T^j - S^i T^j$ is the index anti-symmetrization of $S \otimes T$. To reach the expression above, we used \eqref{2 Theta on coordinates} and the defining property \eqref{4 def adjugate} of the adjugate matrix, taking into account the two different $\mathfrak{sl}_n$ sectors of $\mathfrak{psl}(n|n)_0$.
\par\noindent
Let $T\indices*{^i^j_I_J}$ be the tensor in parenthesis, so that
\begin{equation*}
    Q_{\operatorname{ad}}(\mu^+)
        = T\indices*{^i^j_I_J} \, \gamma\indices{^J_j} \otimes \tau\indices{_i^I} 
            + \det(\gamma^+) \, \mathbb{1}_{2n} ;
\end{equation*}
by its defining expression, $T$ has to be totally anti-symmetric in its upper indices and lower indices, hence $T\indices*{^i^j_I_J} = T \, \epsilon\indices{^i^j} \epsilon\indices{_I_J}$, for some scalar $T$ (that depends on $a$ and $b$). Thus
\begin{align*}
    Q_{\operatorname{ad}}(\mu^+)
        = T \mu^+ + \det(\gamma^+) \, \mathbb{1}_{2n} .
\end{align*}
It is easy to check that
\begin{alignat*}{4}
    i &\neq j
        \quad &&\Longrightarrow \quad
        \epsilon\indices{^k^{[i}} \, a\indices{^{j]}_k}
        &&= \epsilon\indices{^j^i} \operatorname{tr}(a) ,
        \\
    I &\neq J
        \quad &&\Longrightarrow \quad
        b\indices{^K_{[I}} \, \epsilon\indices{_{J]}_K}
        &&= \epsilon\indices{_J_I} \operatorname{tr}(b) ,
\end{alignat*}
thus
\begin{align*}
    T\indices*{^i^j_I_J}
        = \bigl( \operatorname{tr}(a) - \operatorname{tr}(b) \bigr) \, \epsilon\indices{_I_J} \epsilon\indices{^i^j}
    \quad\Longrightarrow\quad
    T = \operatorname{tr}(a) - \operatorname{tr}(b) .
\end{align*}
For a general $n\in\mathbb{N}$, a straightforward computation with adequate index relabelling yields
\begin{align*}
    Q_{\operatorname{ad}} (\mu^+)
        &= \sum_{s=1}^{n} \Biggl(
            \epsilon\indices{_{J_1}_\ldots_{J_n}}
            \epsilon\indices{^{i_1}^\ldots^{i_{s-1}}^k^{i_{s+1}}^\ldots^{i_n}} a\indices{^{i_s}_k}
            + {} \\ &\qquad {}
            - \epsilon\indices{^{i_1}^\ldots^{i_n}} \epsilon\indices{_{J_1}_\ldots_{J_{s-1}}_K_{J_{s+1}}_\ldots_{J_n}} b\indices{^K_{J_s}} \Biggr) \, \gamma\indices{^{J_1}_{i_1}} \cdots \gamma\indices{^{J_{n-1}}_{i_{n-1}}} \otimes \tau\indices{_{i_n}^{J_n}}
        + {} \\ &\quad {}
        + \det(\gamma^+) \, \mathbb{1}_{2n} ;
\end{align*}
again, we have a totally anti-symmetric tensor $T\indices*{^{i_1}^\ldots^{i_n}_{J_1}_\ldots_{J_n}}$, which has to take the form $T \, \epsilon\indices{^{i_1}^\ldots^{i_n}} \epsilon\indices{_{J_1}_\ldots_{J_n}}$. As before, the summations over $k$ and $K$ will yield
\begin{alignat*}{4}
    i_1 &\neq \cdots \neq i_n
    \quad&&\Longrightarrow\quad
    \sum_{s=1}^{n} \epsilon\indices{_{J_1}_\ldots_{J_n}} \epsilon\indices{^{i_1}^\ldots^{i_{s-1}}^k^{i_{s+1}}^\ldots^{i_n}} \, a\indices{^{i_s}_k}
        &&= \epsilon\indices{^{i_1}^\ldots^{i_n}} \operatorname{tr}(a)
    \\
    J_1 &\neq \cdots \neq J_n
    \quad&&\Longrightarrow\quad
    \sum_{s=1}^{n} \epsilon\indices{^{i_1}^\ldots^{i_n}} \epsilon\indices{_{J_1}_\ldots_{J_{s-1}}_K_{J_{s+1}}_\ldots_{J_n}} \, b\indices{^K_{J_s}}
        &&= \epsilon\indices{_{J_1}_\ldots_{J_n}} \operatorname{tr}(b) ,
\end{alignat*}
so that again
\begin{equation*}
    T = \operatorname{tr}(a) - \operatorname{tr}(b) ,
\end{equation*}
therefore
\begin{equation*}
    Q_{\operatorname{ad}}(\mu^+)
        = \bigl( \operatorname{tr}(a) - \operatorname{tr}(b) \bigr) \, \mu^+ + \det(\gamma^+) \, \mathbb{1}_{2n} .
\end{equation*}
This is the image of $\mu^+$ under the BRST operator of $\mathfrak{gl}(n|n)$. By going to $\mathfrak{sl}(n|n)$, we remove the first term (as the supertrace must vanish), and by quotienting out the ideal generated by $\mathbb{1}_{2n}$, nothing is left\footnotemark. Hence $\mu^+$ is closed.
\end{proof}

\footnotetext{
    The last step in the proof of \Cref{prop: adjugates are Q-closed for psl} tells us that the adjugate covariants $\mu^\pm$ are not closed for $\mathfrak{sl}(n|n)$, since $Q_{\operatorname{ad}}^{\mathfrak{sl}}(\mu^+) = \det(\gamma^+) \, \mathbb{1}_{2n} \neq 0$.
}

We can summarize \Cref{prop: BRST on Gamma pm,,prop: BRST on pi M and pi N,prop: adjugates are Q-closed for psl} as:
\begin{alignat*}{3}
    \left.
    \begin{aligned}
        \pi_0(M^s) + \pi_1(N^s) \\
        \Gamma_s^+ - \Gamma_s^-
    \end{aligned}
    \right\}
        &\in \operatorname{im}(Q_{\operatorname{ad}}^{\mathfrak{psl}}) ,
        \quad &&\forall s\in\mathbb{N} ,
    \\
    \left.
    \begin{aligned}
        \pi_0(M^s) - \pi_1(N^s) \\
        \Gamma_s^+ + \Gamma_s^-
    \end{aligned}
    \right\}
        &\notin \ker(Q_{\operatorname{ad}}^{\mathfrak{psl}}) ,
        \quad &&\forall s\in\mathbb{N} ,
    \\
    \left.
    \begin{aligned}
        \Gamma^+_0 - \Gamma_0^- \\
        \mu^+ ,\ \mu^-
    \end{aligned}
    \right\}
        &\in \ker(Q_{\operatorname{ad}}^{\mathfrak{psl}}) / \operatorname{im}(Q_{\operatorname{ad}}^{\mathfrak{psl}}) .
\end{alignat*}

\subsection[From the \texorpdfstring{\nth{2}}{2nd} to the \texorpdfstring{$\infty^{\text{th}}$}{infinitieth} page of the HS spectral sequence]{The \texorpdfstring{\nth{2}}{2nd} to the \texorpdfstring{$\infty^{\text{th}}$}{infinitieth} page of the Hochschild--Serre spectral sequence} \label{subsec: 2nd page HS adjoint}

From \eqref{2 dr schematic} and \eqref{2 delta schematic}, we know what $d_1$ must be:
\begin{align*}
    d_1 \simeq j^\ast \, i^{-1} \, Q \, j^{-1} ,
    \quad\text{up to cohomology classes} ;
\end{align*}
in a diagram, we are interested in the path
\begin{equation*}
  \begin{tikzcd}
    E_1^{p,k-p} \arrow{r}{\Delta}
        & H^{k+1}(F^{p+1} \mathcal{C}) \arrow{r}{j^\ast}
        & E_1^{p+1,k-p} ,
  \end{tikzcd}
\end{equation*}
where $\Delta \cong u^{-1} \, Q \, j^{-1}$ (up to cohomology classes) is the connecting homomorphism. The differential $d_1$ is the map that moves horizontally along $E_1^{\bullet,\bullet}$.

Starting at the $E_1^{p,0} \cong \operatorname{Inv}_{\mathfrak{g}_0} \bigl( \mathbb{K}[U]^p \otimes \mathfrak{g} \bigr)$ row, we have
\begin{equation*}
  \begin{tikzcd}
    E_1^{p,0} \arrow{r}{\Delta}
        & H^{p+1}(F^{p+1} \mathcal{C}) \arrow{r}{j^\ast}
        & E_1^{p+1,0} ;
  \end{tikzcd}
\end{equation*}
since $Q(\phi)$ has the same number of bosonic ghosts as $\phi\in\mathbb{K}[U]$, the maps $i$ and $j$ (in the definition of $\Delta$) and the homomorphism $j^\ast$ effectively act as identities, thus $d_1 \simeq Q$. We then have
\begin{equation*}
    \left.
    \begin{matrix}
        \gamma^+ - \gamma^- \\
        \mu^+ ,\ \mu^-
    \end{matrix}
    \right\}
    \in \ker(d_1) / \operatorname{im}(d_1) = E_2^{p,q} .
\end{equation*}
The reason why $\gamma^+ - \gamma^- \notin \operatorname{im}(d_1)$ is obvious: there is nothing at $E_1^{0,0}$ that could be its pre-image under $d_1$! For $\mu^\pm$, we note the following.

\begin{proposition} \label{prop: upsilon charge}
The even supervector field $\Upsilon \in \operatorname{Vect}(\Pi\mathfrak{g})$, associated to the dual of the map $\operatorname{id}_{U^+} - \operatorname{id}_{U^-} \colon \mathfrak{g} \longrightarrow \mathfrak{g}$, is given in coordinates by
\begin{align*}
    \Upsilon
        \coloneqq \gamma\indices{^I_i} \, \frac{\partial}{\partial\gamma\indices{^I_i}}
        - \gamma\indices{^j_J} \, \frac{\partial}{\partial\gamma\indices{^j_J}}
\end{align*}
and commutes with the BRST operator of $\mathfrak{gl}(m|n)$ in trivial coefficients.
\end{proposition}
\begin{proof}
By definition, when restricted to \nth{1} order monomials, $\Upsilon$ acts as the dual of the map $\operatorname{id}_{U^+} - \operatorname{id}_{U^-} \equiv \gamma^+ - \gamma^-$, hence it commutes with the Chevalley--Eilenberg differential. Furthermore, in matrix form, we have
\begin{align*}
    \Upsilon Q (\phi)
        &= \gamma^+ \frac{\partial}{\partial\gamma^+}\bigl( Q(\phi) \bigr)
            - \gamma^- \frac{\partial}{\partial\gamma^-} \bigl( Q(\phi) \bigr) ,
    \\
    Q \Upsilon(\phi)
        &= Q(\gamma^+) \frac{\partial\phi}{\partial\gamma^+}
            + \gamma^+ Q\biggl(\frac{\partial\phi}{\partial\gamma^+}\biggr)
            - Q(\gamma^-) \frac{\partial\phi}{\partial\gamma^-}
            - \gamma^- Q\biggl(\frac{\partial\phi}{\partial\gamma^-}\biggr) ;
\end{align*}
but
\begin{align*}
    \biggl[ Q , \frac{\partial}{\partial c\indices{^A}} \biggr]
        &\equiv Q \, \frac{\partial}{\partial c\indices{^A}}
            + (-1)^{\bar{A}} \, \frac{\partial}{\partial c\indices{^A}} Q
        = f\indices*{^D_B_A} \, c\indices{^B} \frac{\partial}{\partial c\indices{^D}} ,
\end{align*}
so that
\begin{align*}
    Q\Upsilon - \Upsilon Q
        &= Q(\gamma^+) \frac{\partial}{\partial\gamma^+}
            - Q(\gamma^-) \frac{\partial}{\partial\gamma^-}
        + {} \\ &\qquad{}
        + \Bigl( f^D_{B\alpha^+} \gamma^+ c^B
            - f^D_{B\alpha^-} \gamma^- c^B \Bigr) \frac{\partial}{\partial c^D} ,
\end{align*}
where the index $\alpha^\pm$ labels the basis of $U^\pm$. For $D$ ranging in $\mathfrak{g}_0$ (i.e. $\bar{D} = 0$), the term in parentheses vanishes; for $D$ ranging in $\mathfrak{g}_1 = U$, the second line cancels the first. Therefore $[Q,\Upsilon]=0$.
\end{proof}

As an operator in $\operatorname{Pol}^\bullet(\Pi\mathfrak{g})$, the spectrum of $\Upsilon$ is $\mathbb{Z}$, which induces a $\mathbb{Z}$-grading we shall call \emph{$\upsilon$-charge} (``upsilon charge''). Its interpretation is rather easy: the $\upsilon$-charge of a monomial is the number of $\gamma^+$'s minus the number of $\gamma^-$'s; the $\upsilon$-charge of $\gamma^\pm$ is $\pm1$, and of $a$ or $b$ is $0$. 

\Cref{prop: upsilon charge} asserts that $Q^\mathfrak{gl}$ in trivial coefficients preserves the $\upsilon$-charge. This is an alternative way to see that the invariant polynomials $\delta^\pm$ \emph{cannot} be in the image of $Q^{\mathfrak{psl}}$ --- and, by extension, of any $d_r$ in \Cref{sec: trivial coefficients}.

\bigskip

For a general representation, however, we have to correct $\Upsilon$ as follows.

\begin{proposition}
Consider the even supervector field
\begin{equation*}
    \Upsilon_\rho
        \coloneqq \Upsilon \otimes \operatorname{id} + \operatorname{id} \otimes \tfrac{1}{2} \rho(\mathbb{J}) 
        \quad \in \operatorname{Vect}(\Pi\mathfrak{g}) \otimes V,
\end{equation*}
where $\rho$ is the $\mathfrak{g}$-action on $V$, and $\mathbb{J} = (+\mathbb{1}_m,-\mathbb{1}_n)$ is the element of $\mathfrak{gl}(m|n)$ defined in \eqref{2 def J generator of supertrace}. Then $\Upsilon_\rho$ commutes with the BRST operator $Q_\rho$ of $\mathfrak{gl}(m|n)$.
\end{proposition}
\begin{proof}
For an arbitrary $V$-valued ghost polynomial $\phi \otimes v \in \operatorname{CE}^\bullet_\rho (\mathfrak{g},V)$, we have
\begin{align*}
    [\Upsilon_\rho , Q_\rho] (\phi\otimes v)
        &= [\Upsilon,Q](\phi) \otimes v + {} \\
        &\quad{} +
            (-1)^{\deg(\phi)} \, \phi \Bigl( 
                \Upsilon(c^A) \otimes \rho(\tau_A)
                - c^A \otimes \tfrac{1}{2} \rho\bigl([\tau_A,\mathbb{J}]\bigr) \Bigr) \cdot v ,
\end{align*}
and \Cref{prop: upsilon charge} already guarantees the first line vanishes for every scalar-valued ghost polynomial $\phi$.
\par\noindent
From the definition of $\Upsilon$, we have
\begin{align*}
    \Upsilon(c^A) \otimes \rho(\tau_A)
        = c^{\alpha^+} \otimes \rho(\tau_{\alpha^+})
        - c^{\alpha^-} \otimes \rho(\tau_{\alpha^-}) ,
\end{align*}
where $\alpha^\pm$ labels the basis of $U^\pm \subset \mathfrak{g}_1$. In turn, we know that 
\begin{align*}
    [\mathfrak{g}_0 , \mathbb{J}] = 0
    \qquad\text{and}\qquad
    [\tau_{\alpha^\pm},\mathbb{J}] = \pm 2\tau_{\alpha^\pm} ,
    \qquad\forall\alpha^\pm \in\{1,\ldots,\dim U^\pm\},
\end{align*}
hence
\begin{equation*}
    c^A \otimes \tfrac{1}{2} \rho\bigl( [\tau_A,\mathbb{J}] \bigr)
        = c^{\alpha^+} \otimes \rho(\tau_{\alpha^+})
        - c^{\alpha^-} \otimes \rho(\tau_{\alpha^-}) ,
\end{equation*}
i.e. $[\Upsilon_\rho,Q_\rho] = 0$.
\end{proof}

Although $\mathbb{J}$ is not present in $\mathfrak{psl}(n|n)$ --- in fact, $\mathbb{J}$ is the element removed from $\mathfrak{gl}(n|n)$ when restricting to $\mathfrak{sl}(n|n)$ ---, we can work with $Q_\rho$ for $\mathfrak{gl}(n|n)$ and see that there is a preserved $\upsilon$-charge then. For $\rho = \operatorname{ad}$, we have that
\begin{itemize}
    \item the adjugate covariants $\mu^\pm$ have $\upsilon = \pm(m-1)$,

    \item the odd-degree trace-like covariants $\Gamma_s^\pm$ have $\upsilon = \pm1$,

    \item the even-degree trace-like covariants $\pi_0(M^s)$ and $\pi_1(N^s)$ have $\upsilon = 0$,

    \item the trace-like invariant polynomials $\eta_{2k}$ have $\upsilon = 0$,

    \item the determinant polynomials $\delta^\pm$ have $\upsilon = \pm n$.
\end{itemize}

The preservation of the $\upsilon$-charge by $Q_{\operatorname{ad}}$ of $\mathfrak{gl}(n|n)$ guarantees, then, that no covariant can ever reach $\mu^\pm$ by $Q_{\operatorname{ad}}$.

\bigskip

We then come to the second main theorem of this paper:

\begin{theorem} \label{thm: psl adjoint coefficients}
The cohomology of $\mathfrak{psl}(n|n)$ in adjoint coefficients is
\begin{align*}
    H^\bullet_{\operatorname{ad}} \bigl( \mathfrak{psl}(n|n) \bigr)
    \cong \mathcal{K}\bigl\{ (\gamma^+ - \gamma^-) ,\ \mu^\pm \bigr\} \big/ \mathcal{J}_{\operatorname{ad}} ,
\end{align*}
i.e. $\mathcal{K}$-linear combinations of the pure covariants $\gamma^+ - \gamma^-$ and $\mu^\pm$, where
\begin{equation*}
    \mathcal{K}
        \coloneqq H\bigl( \mathfrak{psl}(n|n) ; \mathbb{K} \bigr) .
\end{equation*}
quotiented by the ideal $\mathcal{J}_{\operatorname{ad}}$ is generated by the relation
\begin{align} \label{4 covariants nontrivial relation} 
    \delta^+ \, \mu^- - \delta^- \, \mu^+
        &= \frac{1}{n!} \, (\eta_2)^n \bigl( \gamma^+ - \gamma^- \bigr) ,
\end{align}
\end{theorem}
\begin{proof}
From the previous discussion, we know that $\gamma^+ - \gamma^-$ and the adjugates $\mu^\pm$ are the only ``pure'' covariants that are not in the image of $d_1 \simeq Q$.
\par\noindent
In a general spot of the \nth{1} page, we have \eqref{thm: Hochschild Serre}, and a general element is of the form $\psi \wedge \phi$, with $\psi\in H^q(\mathfrak{g}_0) = H^q(\mathfrak{sl}_n \oplus \mathfrak{sl}_n)$ and $\phi$ a $p^{\text{th}}$-order covariant, thus
\begin{align*}
    Q_{\operatorname{ad}}(\psi \wedge \phi)
        = Q(\psi) \wedge \phi
            - \psi \wedge Q_{\operatorname{ad}}(\phi) .
\end{align*}
At the \nth{2} page, the second term always vanishes, since all surviving covariants ($\gamma^+ - \gamma^-$, $\delta^\pm$ and their $\mathcal{R}$-linear combinations) stay in the kernel of $Q$.
\par\noindent
From trivial coefficients, we know that $\psi$ has to be of the type \eqref{2 fermionic generators}, which transgress to $\eta_{2k} \equiv \operatorname{tr}(M^k) = \operatorname{tr}(N^k)$. This means that, for $k>1$, 
\begin{equation*}
    \eta_{2k} \wedge \phi
        = d_{2k}^{} \bigl( \varphi'_{2k-1} \wedge \phi \bigr) .
\end{equation*}
The only invariant polynomials left are $\eta_2$ and $\delta^\pm$, with the restriction \eqref{3 nontrivial relation cohomology}. Therefore, we know that the cohomology in adjoint coefficients is a module over $\mathcal{K} = H\bigl( \mathfrak{psl}(n|n) ; \mathbb{K} \bigr)$ --- described in \Cref{thm: psl trivial coefficients}.
\par\noindent
Finally, to prove \eqref{4 covariants nontrivial relation}, we make use of \Cref{prop delta mu}:
\begin{align*}
    (-1)^n \bigl( \delta^+ \mu^- - \delta^- \mu^+ \bigr)
        &= \bigl( \Gamma\indices*{^-_{n-1}} - \Gamma\indices*{^+_{n-1}} \bigr)
            + {} \\ &\qquad {} 
            + c_{n-1} \bigl( \Gamma\indices*{^-_{n-2}} - \Gamma\indices*{^+_{n-2}} \bigr)
            + \cdots 
            + c_1 \bigl( \Gamma\indices*{^-_0} - \Gamma\indices*{^+_0} \bigr) ,
\end{align*}
because the coefficients $c_s$ from the characteristic polynomials of $M$ and $N$ are the same. As we know from \Cref{prop: BRST on pi M and pi N}, all $\Gamma^-_s - \Gamma^+_s$ are exact for $s\neq0$, hence
\begin{equation*}
    \delta^+ \mu^- - \delta^- \mu^+
        = \frac{1}{(n-1)!} (\eta_2)^{n-1} \bigl( \gamma^+ - \gamma^- \bigr)
        \qquad\text{in cohomology},
\end{equation*}
as we wanted.
\end{proof}

\subsubsection{Cohomology of \texorpdfstring{$\mathfrak{psl}(2|2)$}{psl(2|2)} in adjoint coefficients}

Let us look at $\mathfrak{psl}(2|2)$ as an example. Its cohomology in adjoint coefficients is
\begin{align*}
    H^\bullet_{\operatorname{ad}}\bigl( \mathfrak{psl}(2|2) \bigr)
    = H^\bullet(\mathfrak{sl}_2)
        \otimes \frac{\mathbb{K}[ \eta_2 ,\ \delta^\pm ]}{\bigl( \delta^+ \delta^- = \tfrac{1}{2} (\eta_2)^2 \bigr)}
        \otimes \frac{\mathbb{K}\bigl\{ (\gamma^+ - \gamma^-) , \mu^\pm \bigr\}}{\mathcal{J}_{\operatorname{ad}}} .
\end{align*}
More explicitly, in \Cref{table: psl 2 2 adjoint} we present the generators up to \nth{7} cohomology, classified by their $\upsilon$-charge.

\begin{table}[H]
\centering
\begin{tabular}{|c|*{4}{>{\arraybackslash}m{2.3cm}|}c|}
    \hline\hline
    \multirow{2}{0.25cm}{\centering$\boldsymbol{k}$}
        & \multicolumn{5}{c|}{ \textbf{generators of} $\boldsymbol{H^k\bigl( \mathfrak{psl}(2|2) ; \operatorname{ad} \bigr)}$ }
    \\ 
        & \centering $\boldsymbol{\upsilon=0}$
        & \centering $\boldsymbol{\upsilon=\pm2}$
        & \centering $\boldsymbol{\upsilon=\pm4}$
        & \centering $\boldsymbol{\upsilon=\pm6}$
        & 
    \\ \hline\hline
    1   & $\bullet\; [\mathbb{J},\gamma]$
        & $\bullet\; \mu^\pm$
        & & &
    \\ \hline
    2   & & & & &
    \\ \hline
    3   & $\bullet\; \eta_2 [\mathbb{J},\gamma]$
            \newline
            $\bullet\; \delta^+ \mu^-$
        & $\bullet\; \eta_2 \mu^\pm$
            \newline
            $\bullet\; \delta^\pm [\mathbb{J},\gamma]$
        & $\bullet\; \delta^\pm \mu^\pm$
        & &
    \\ \hline
    4   & $\bullet\; \psi_3 [\mathbb{J},\gamma]$
        & $\bullet\; \psi_3 \mu^\pm$
        & & &
    \\ \hline
    5   & $\bullet\; (\eta_2)^2 [\mathbb{J},\gamma]$ 
            \newline
            $\bullet\; \eta_2 \delta^+ \mu^-$
        & $\bullet\; \eta_2 \delta^\pm [\mathbb{J},\gamma]$
            \newline
            $\bullet\; (\eta_2)^2 \mu^\pm$
        & $\bullet\; (\delta^\pm)^2 [\mathbb{J},\gamma]$
            \newline
            $\bullet\; \eta_2 \delta^\pm \mu^\pm$
        & $\bullet\; (\delta^\pm)^2 \mu^\pm$
        &
    \\ \hline
    6   & $\bullet\; \psi_3 \eta_2 [\mathbb{J},\gamma]$
            \newline
            $\bullet\; \psi_3 \delta^+ \mu^-$
        & $\bullet\; \psi_3 \delta^\pm [\mathbb{J},\gamma]$
            \newline
            $\bullet \; \psi_3 \eta_2 \mu^\pm$
        & $\bullet \; \psi_3 \delta^\pm \mu^\pm$
        & &
    \\ \hline
    7   & $\bullet\; (\eta_2)^3 [\mathbb{J},\gamma]$
        & $\bullet\; (\eta_2)^2 \mu^\pm [\mathbb{J},\gamma]$
            \newline
            $\bullet\; (\eta_2)^3 \mu^\pm$
            \newline
            $\bullet\; (\eta_2)^2 \delta^+ \mu^-$
        & $\bullet\; \eta_2 (\delta^\pm)^2 [\mathbb{J},\gamma]$
            \newline
            $\bullet\; (\eta_2)^2 \delta^\pm \mu^\pm$
        & $\bullet\; (\delta^\pm)^3 [\mathbb{J},\gamma]$
            \newline
            $\bullet\; \eta_2 (\delta^\pm)^2 \mu^\pm$
        & $\star$
    \\ \hline
    \multicolumn{6}{c}{$\vdots$}
\end{tabular}

\begin{flushright}
$\star = \quad$
\begin{tabular}{|c|c|c|}
    \hline\hline
    $\boldsymbol{k}$
        & $\boldsymbol{\cdots}$
        & $\boldsymbol{\upsilon=\pm8}$
    \\ \hline\hline
    7   & $\cdots$ & $\bullet\; (\delta^\pm)^3 \mu^\pm$
    \\ \hline
\end{tabular}
\end{flushright}
\caption{Generators of the low order cohomology spaces of $\mathfrak{psl}(2|2)$ in adjoint coefficients, classified by $\upsilon$-charge. Note that $\frac{1}{2} [\gamma,\mathbb{J}] = \gamma^+ - \gamma^-$.}
\label{table: psl 2 2 adjoint}
\end{table}

We see how quickly the dimension of the cohomology groups of $\mathfrak{psl}(2|2)$ grow. For higher $n$, the cohomology of $\mathfrak{psl}(n|n)$ will contain more fermionic ghost generators $\psi$ and will allow for higher powers of $\eta_2$ (before imposing the non-trivial relation), which enrich the cohomology structure.

\section{Final remarks}

In this paper, we computed explicitly the cohomology of the projective special linear Lie superalgebras $\mathfrak{psl}(n|n)$ in trivial and adjoint coefficients, using their associated Hochschild--Serre spectral sequences.

We showed how the determinants of the odd-blocks in the canonical decomposition \eqref{2 gl decomposition} play an important role in $\mathfrak{sl}(n|n)$ and $\mathfrak{psl}(n|n)$ and, moreover, how there are non-trivial relations between their cohomology classes --- which are described by the second fundamental theorem of invariant theory.

As mentioned in \Cref{sec: intro}, we borrowed the motivation from physics and the AdS/CFT duality, in which the real Lie superalgebra $\mathfrak{psu}(2,2|4)$ appears. Once again, we have the isomorphism in cohomology
\begin{equation*}
    H\bigl( \mathfrak{psl}_{\mathbb{C}}(n|n) ; V_\mathbb{C} \bigr)
        \cong \mathbb{C} \, H\bigl( \mathfrak{psu}(n|n) , V \bigr) ,
\end{equation*}
i.e., the cohomology of $\mathfrak{psl}(n|n)$ over $V_\mathbb{C} = V \otimes_{\mathbb{R}} \mathbb{C}$ is the complexification of the cohomology of $\mathfrak{psu}(n|n)$. Therefore, what we presented here can be extended to $\mathfrak{psu}(n|n)$ once one takes the \emph{real forms} of $\mathfrak{psl}(n|n)$ and its cohomology.

\addtocontents{toc}{\protect\vspace{\beforebibskip}}
\addcontentsline{toc}{section}{Acknowledgements}
\section*{Acknowledgements}

I would like to thank Andrei Mikhailov for the valuable discussions.

I am grateful for the hospitality of Perimeter Institute where part of this work was carried out. Research at Perimeter Institute is supported in part by the Government of Canada through the Department of Innovation, Science and Economic Development and by the Province of Ontario through the Ministry of Colleges and Universities.

    \addtocontents{toc}{\protect\vspace{\beforebibskip}}
    \addcontentsline{toc}{section}{\refname}
    \printbibliography[title={References}]

\end{document}